\documentclass[letterpaper, journal, twoside]{IEEEtran}

\IEEEoverridecommandlockouts 

\usepackage{graphics} 
\usepackage{epsfig} 
\usepackage{amsmath} 
\usepackage{amssymb} 

\usepackage{algorithm}
\usepackage{algorithmicx}
\usepackage{algpseudocode}

\usepackage{epstopdf}

\usepackage{verbatim}

\usepackage{amsthm}

\usepackage{color}
\usepackage{cancel}

\usepackage{enumitem}

\usepackage{tikz} 
\usepgflibrary{arrows}

\usepackage{hyperref}

\usepackage[normalem]{ulem}
\usepackage{enumitem}

\newtheorem{assumption}{\bf Assumption}

\newtheorem{theorem}{\bf Theorem}
\newtheorem{proposition}{\bf Proposition}
\newtheorem{lemma}{\bf Lemma}

\newtheorem{remark}{\bf Remark}

\usepackage{cite}

\usepackage{booktabs}
\usepackage{siunitx}

\usepackage{doi}

\newcommand{\norm}[1]{\left\lVert#1\right\rVert}
\newcommand{\myset}[2]{\left\{#1\,\left|\, #2\right.\right\}}
\newcommand{\nnorm}[1]{\lVert#1\rVert}

\newcommand\copyrighttext{%
  \scriptsize \textcopyright 2026 IEEE. Personal use of this material is permitted.
  Permission from IEEE must be obtained for all other uses, in any current or future
  media, including reprinting/republishing this material for advertising or promotional
  purposes, creating new collective works, for resale or redistribution to servers or
  lists, or reuse of any copyrighted component of this work in other works.}
\newcommand\copyrightnotice{%
\begin{tikzpicture}[remember picture,overlay]
\node[anchor=south,yshift=10pt] at (current page.south) 
  {\fbox{\parbox{\dimexpr\textwidth-\fboxsep-\fboxrule\relax}{\copyrighttext}}};
\end{tikzpicture}%
}

\begin{document}
\title{Stochastic MPC with Online-optimized Policies and Closed-loop Guarantees}
\author{
Marcell Bartos, Alexandre Didier, Jerome Sieber, Johannes K\"ohler$^\dagger$, Melanie N. Zeilinger$^\dagger$
\thanks{$^\dagger$: joint supervision. All authors are with the Institute for Dynamic Systems and Control, ETH Zurich, 8092 Zürich, Switzerland (e-mail: [mbartos, adidier, jsieber, mzeilinger]@ethz.ch, j.kohler@imperial.ac.uk). Johannes K\"ohler has been supported by the Swiss National Science Foundation under NCCR Automation (grant agreement 51NF40\textunderscore180545). This work was supported as a part of NCCR Automation, a National Centre of Competence in Research, funded by the Swiss National Science Foundation (grant number 51NF40\textunderscore225155).}
}

\maketitle
\copyrightnotice
\thispagestyle{empty}
\begin{abstract} 
This paper proposes a stochastic model predictive control method for linear systems affected by additive Gaussian disturbances that optimizes over disturbance feedback matrices online. Closed-loop satisfaction of probabilistic constraints and recursive feasibility of the underlying convex optimization problem is guaranteed. Optimization over feedback policies online increases performance and reduces conservatism compared to fixed-feedback approaches. The central mechanism is a finitely determined maximal admissible set for probabilistic constraints, together with the reconditioning of the predicted probabilistic constraints on the current knowledge at every time step. The proposed method's applicability is demonstrated on a building temperature control example.
\end{abstract}

\section{Introduction} \label{sec:intro}
Model predictive control (MPC)~\cite{rawlings2017model, kouvaritakis2016model} is a powerful approach for the optimal control of constrained systems, by solving finite horizon optimization problems in a receding horizon manner. For deterministic systems, MPC offers strong theoretical guarantees in terms of constraint satisfaction, stability, and performance.
However, ensuring constraint satisfaction for stochastic dynamical systems is challenging. To address this issue, stochastic MPC (SMPC) schemes~\cite{farina2016stochastic, mesbah2016stochastic} consider a probabilistic description of the uncertainty and enforce safety-critical constraints with a user-chosen probability.
In this work, we propose an SMPC scheme for linear time-invariant systems affected by additive Gaussian disturbances that
\begin{enumerate}[label=\emph{\alph*})]
    \item improves performance by using online-optimized affine feedback policies while preserving convexity,\label{req:optimized_feedback}
    \item and guarantees recursive feasibility and satisfaction of the chance constraints in closed loop.\label{req:guarantees}
\end{enumerate}

\subsubsection*{Related work} 
If the support of the distribution of the disturbance is bounded, recursive feasibility of SMPC optimization problems can be ensured by adopting robust constraint tightening approaches, see, e.g.,~\cite{cannon2010stochastic, lorenzen2016constraint}. However, in the unbounded support case (such as the Gaussian distribution considered in this paper), guaranteeing recursive feasibility becomes nontrivial~\cite{farina2016stochastic, ono2012joint}.
A common solution is to use recovery mechanisms~\cite{farina2013probabilistic}, or soften the constraints~\cite{paulson2020stochastic}, but this often leads to the loss of closed-loop chance constraint satisfaction guarantees, unless stronger assumptions are used~\cite{hewing2018stochastic, kohler2022recursively, schluter2022stochastic}.
In contrast, the methods of~\cite{hewing2020recursively, arcari2023stochastic, mark2024stochastic} achieve property~\ref{req:guarantees} via the so-called indirect feedback paradigm. These methods do not initialize the predicted trajectories that enter the constraints of the optimization problem based on the current measurement of the state, i.e., they do not \emph{recondition} on current knowledge. Instead, they use the shifted optimal nominal state from the previous solution for initialization, and use a separate predicted trajectory for the cost function, which is initialized by the true state.
All of the aforementioned methods that achieve~\ref{req:guarantees} rely on offline-designed feedback policies to construct probabilistic reachable sets, and thus they do not offer~\ref{req:optimized_feedback}, which limits their performance.
\par As for requirement~\ref{req:optimized_feedback}, the advantages of optimizing over the feedback in terms of improved performance and reduced conservatism have been shown in numerous prior works, including~\cite{lofberg2003approximations, oldewurtel2008tractable, sieber2021system, leeman2025robust}. Optimization over linear state feedback policies in MPC can be written as a convex problem using a parameterization in terms of disturbance feedback~\cite{goulart2006optimization} or equivalently by using the system level parameterization~\cite{anderson2019system}.
Corresponding SMPC schemes with online optimized policies have been proposed in~\cite{oldewurtel2008tractable, prandini2012randomized, mark2022recursively, pan2023data, li2024distributionally, knaup2024recursively}.
However,~\cite{oldewurtel2008tractable} can only guarantee recursive feasibility for bounded disturbances,~\cite{prandini2012randomized} does not provide recursive feasibility guarantees, and~\cite{mark2022recursively, pan2023data, li2024distributionally, knaup2024recursively} utilize a recovery mechanism to ensure recursive feasibility, but they do not guarantee chance constraint satisfaction for the closed-loop system.
That is, while these methods achieve~\ref{req:optimized_feedback}, they do not achieve~\ref{req:guarantees}. 
\par Overall, to the authors' best knowledge, there does not exist an SMPC scheme in the literature that can address requirements~\ref{req:optimized_feedback}--\ref{req:guarantees} simultaneously.
Furthermore, existing SMPC schemes with closed-loop guarantees cannot be directly extended to use online-optimized feedback.
In particular, the constraint satisfaction guarantees of~\cite{hewing2020recursively, arcari2023stochastic, mark2024stochastic} heavily rely on the independent evolution of the closed-loop error from the optimization problem, while unimodality of its distribution is utilized in~\cite{hewing2018stochastic, kohler2022recursively, schluter2022stochastic}. In case of receding horizon SMPC using online-optimized feedback, the feedback matrices become random variables, meaning that the error is neither unimodal nor independent from the optimization problem. This highlights that achieving~\ref{req:optimized_feedback}--\ref{req:guarantees} simultaneously in the presence of probabilistic constraints and unbounded disturbances poses a significant challenge. Note that this a fundamental difference compared to robust MPC approaches, where extending fixed feedback methods to use online-optimized feedback while retaining the desired closed-loop guarantees is straightforward~\cite{goulart2006optimization}.

\subsubsection*{Contribution}
This paper proposes an SMPC scheme for linear systems affected by additive Gaussian disturbances that uses online-optimized feedback policies, while ensuring recursive feasibility, closed-loop chance constraint satisfaction, and bounded asymptotic average cost. Despite the presence of unbounded disturbances, we guarantee recursive feasibility by appropriately relaxing the constraints in the optimization problem using the idea of reconditioning, while still satisfying the original probabilistic constraints in closed loop. This is achieved via three technical contributions.
\begin{enumerate}
    \item \label{contr:1} We propose a general receding horizon reconditioning SMPC formulation, based on the shrinking horizon formulation of~\cite{wang2021recursive}, where the probabilistic constraints are reconditioned on all past disturbances at every time step.
    \item \label{contr:2} Inspired by~\cite{li2021chance}, we construct a finitely determined maximal admissible set for probabilistic constraints that can function as the terminal set of SMPC methods that utilize online-optimized feedback in the prediction horizon.
    \item \label{contr:3} We formulate a tractable SMPC scheme for linear systems affected by additive i.i.d. Gaussian disturbances that optimizes over linear causal feedback policies online (achieving~\ref{req:optimized_feedback}), while also ensuring~\ref{req:guarantees}.
\end{enumerate}

\subsubsection*{Outline}
Section~\ref{sec:preliminaries} contains the problem setup and preliminaries on convex optimization over feedback policies, while Section~\ref{sec:concept} introduces the conceptual proposed SMPC scheme and provides a roadmap for the rest of the paper. In Section~\ref{sec:DMC}, we formulate a finite horizon optimal control problem that yields a sequence of optimized policies that guarantee chance constraint satisfaction for all times. Then, we present a general receding horizon SMPC framework using the concept of reconditioning in Section~\ref{sec:reconditioning}. Building on the previous two sections, Section~\ref{sec:RHC} contains the proposed receding horizon scheme that re-solves the optimization problem at every time step, along with a theoretical closed-loop analysis. Then, the proposed method's applicability is demonstrated on a building temperature control example in Section~\ref{sec:numerical}, and Section~\ref{sec:conclusion} concludes the paper. Some proofs and auxiliary results can be found in the Appendix.

\section{Setup and Preliminaries} \label{sec:preliminaries}
After a summary of the used notation, the problem setup is presented, which is followed by a brief summary of relevant results from system level synthesis.

\subsection{Notation}
The set of positive (non-negative) reals is denoted by $\mathbb{R}_{> 0}$ ($\mathbb{R}_{\geq 0}$), the set of integers bigger (not smaller) than $a \in \mathbb{R}$ is denoted by $\mathbb{I}_{> a}$ ($\mathbb{I}_{\geq a}$), and the set of integers in the interval $[a, b]$ is denoted by $\mathbb{I}_{[a, b]}$. The interior of set $\mathcal{A}$ is denoted by $\mathrm{int}(\mathcal{A})$.
$I_n$ represents the $n\times n$ identity matrix, positive (semi-)definiteness of a matrix $A$ is indicated by $A\succ(\succeq)\ 0$, and the Kronecker product of two matrices is denoted by $\otimes$. For a square matrix $A$, $\mathrm{tr}(A)$ and $\rho(A)$ denote its trace and spectral radius, respectively. The 2-norm of vector $x$ is denoted by $\norm{x}$, while $\norm{x}_P := \sqrt{x^\top P x}$ for $P = P^\top \succeq 0$. For a real matrix $B$, $\norm{B} := \sqrt{\rho(B^\top B)}$. For $0 \preceq \Sigma = \Sigma^\top \in \mathbb{R}^{n\times n}$, $\Sigma^{1/2}$ is the unique matrix that satisfies $\Sigma^{1/2} = (\Sigma^{1/2})^\top \succeq 0,\ \Sigma^{1/2}\Sigma^{1/2} = \Sigma$. For a sequence of vectors $\{a_i\}$, $a_{j:k}$ denotes the stacked vector $\begin{bmatrix} a_j^\top & ... & a_k^\top\end{bmatrix}^\top$, and we use the convention that $\sum_{i=j}^k a_i = 0$, whenever $k<j$. The vectorization of a matrix $B \in \mathbb{R}^{n\times m}$ is denoted by $\mathrm{vec}(B) \in \mathbb{R}^{nm}$, which is obtained by stacking the columns of $B$ on top of one another. We sometimes use the abbreviation $XAX^\top = XA[*]^\top$ for symmetric matrix expressions. If $x$ is a decision variable in an optimization problem, $x^\star$ denotes its optimal value.
The probability of some event $A$ and its conditional probability conditioned on event $B$ is denoted by $\mathrm{Pr}[A]$ and $\mathrm{Pr}[A | B]$, respectively. Similarly, for random variables $x$ and $y$, $\mathbb{E}[x]$, $\mathbb{E}[x|y]$, $\mathrm{Var}[x]$, $\mathrm{Var}[x|y]$, $p(x)$, $p(x|y)$ refer to the (conditional) expected value, variance, and probability density function, respectively. The multivariate normal distribution with mean $\mu \in \mathbb{R}^n$ and variance $0 \preceq \Sigma = \Sigma^\top \in \mathbb{R}^{n\times n}$ is denoted by $\mathcal{N}(\mu, \Sigma)$, while $\chi^2(\cdot)$ denotes the inverse cumulative distribution function of the chi-squared distribution with a single degree of freedom.

\subsection{Problem Setup}
We consider a discrete-time linear time-invariant system that is affected by additive independent and identically distributed (i.i.d.) zero-mean Gaussian disturbances with variance $\Sigma^\mathrm{w}\succ~0$:
\begin{equation} \label{eq:system}
    x_{k+1} = A x_k + B u_k + w_k, \quad w_k~\sim~\mathcal{N}(0, \Sigma^\mathrm{w})
\end{equation}
with state $x_k \in \mathbb{R}^n$, input $u_k \in \mathbb{R}^m$, disturbance $w_k \in \mathbb{R}^n$, time $k \in \mathbb{I}_{\geq 0}$, and stabilizable pair $(A,B)$. We assume that the state can be directly measured. The system is subject to half-space chance constraints
\begin{equation} \label{eq:chance_constraints}
    \mathrm{Pr}[(x_k, u_k) \in \mathcal{C}^j] \geq p^j,\ \forall j \in \mathbb{I}_{[1,c]},\ \forall k \in \mathbb{I}_{\geq 0},
\end{equation}
with user-defined probability $p^j \in (0.5,1)$ and
\begin{equation*}
    \mathcal{C}^j = \myset{(x,u) \in \mathbb{R}^n \times \mathbb{R}^m}{G_jx + H_ju \leq b_j },
\end{equation*}
where $G_j$, $H_j$, and $b_j$ denote the $j$-th row of $G \in \mathbb{R}^{c\times n}$, $H \in \mathbb{R}^{c\times m}$, and $b \in \mathbb{R}^{c}_{> 0}$, respectively. The user-defined probability levels $p^j$ allow for a trade-off between performance and safety, although there is a maximal level above which the problem is infeasible (cf. Assumption~\ref{ass:L}), in which case $p^j$ needs to be decreased. Note that the proposed method is unable to handle hard input constraints, which is a general limitation of constrained stochastic control methods in the presence of disturbances with unbounded support.

Consider the linear-quadratic stage cost
\begin{equation} \label{eq:ell}
    \ell(x, u) = \norm{x}^2_Q + q^\top x+\norm{u}^2_R+r^\top u,
\end{equation}
where $q \in \mathbb{R}^n$, $r \in \mathbb{R}^m$, $Q \in\mathbb{R}^{n\times n}$, $R \in\mathbb{R}^{m\times m}$, and $Q, R \succeq 0$, and $\ell$ admits a uniform lower bound. The goal is to satisfy~\eqref{eq:chance_constraints} while minimizing the expected cumulative cost~\eqref{eq:ell} over an infinite horizon. The corresponding stochastic optimal control problem is intractable in general~\cite{kouvaritakis2016model} due to the presence of an infinite number of optimization variables and constraints, and due to the fact that it involves optimizing over arbitrary causal control policies.

In order to overcome these challenges, we restrict the considered policy class to disturbance affine feedback laws (see Section~\ref{sec:SLP}), and propose a stochastic MPC scheme that approximates the stochastic optimal control problem by repeatedly solving a finite horizon optimal control problem in a receding horizon manner (see Section~\ref{sec:concept}).

\subsection{Optimized Feedback via System Level Parameterization} \label{sec:SLP}
In the following, we review how the system level parameterization~\cite{anderson2019system} (or equivalently disturbance feedback~\cite{goulart2006optimization}) has been used in the literature to propagate uncertainty and optimize feedback policies in a convex manner.
\par Consider a finite horizon with length $N \in \mathbb{I}_{> 0}$. In order to address the presence of uncertainty, we employ a common strategy in MPC by separating the system into nominal and error dynamics. Define the nominal system that evolves in a disturbance-free manner as
\begin{equation*}
    z_{k+1} = A z_k + B v_k,
\end{equation*}
initialized at $z_0 = x_0$, where $z_k \in \mathbb{R}^n$ and $v_k \in \mathbb{R}^m$ denote the nominal state and input at time $k$, respectively. After defining the error system
\begin{equation*}
    e^\mathrm{x}_{k+1} = A e^\mathrm{x}_k + B e^\mathrm{u}_k + w_k,\  e^\mathrm{x}_0 = 0,
\end{equation*}
where $x_k = z_k + e^\mathrm{x}_k$ and $u_k = v_k + e^\mathrm{u}_k$, we consider the feedback parameterization
\begin{equation} \label{eq:e^u_k}
    e^\mathrm{u}_k = \sum_{i=1}^k \Phi^\mathrm{u}_{k,i} w_{i-1},
\end{equation}
in which case the resulting error trajectory can be expressed as
\begin{equation} \label{eq:e^x_k}
    e^\mathrm{x}_k = \sum_{i=1}^k \Phi^\mathrm{x}_{k,i} w_{i-1}.
\end{equation}
The system response $\Phi^\mathrm{x}_{k,i}$ evolves according to
\begin{equation} \label{eq:SLP}
    \mathbf{\Phi}^\mathrm{x}_1 = I_n,\quad \mathbf{\Phi}^\mathrm{x}_{k+1} = \begin{bmatrix}
        A \mathbf{\Phi}^\mathrm{x}_k + B \mathbf{\Phi}^\mathrm{u}_k && I_n
    \end{bmatrix},\ \forall k \in \mathbb{I}_{[1,N-1]},
\end{equation}
with block matrices
\begin{align*}
    \mathbf{\Phi}^\mathrm{x}_k &= \begin{bmatrix} \Phi^\mathrm{x}_{k,1} & ... & \Phi^\mathrm{x}_{k,k}\end{bmatrix} \in \mathbb{R}^{n \times kn},\\ 
    \mathbf{\Phi}^\mathrm{u}_k &= \begin{bmatrix} \Phi^\mathrm{u}_{k,1} & ... & \Phi^\mathrm{u}_{k,k}\end{bmatrix} \in \mathbb{R}^{m \times kn}.
\end{align*}

Constraint~\eqref{eq:SLP} is commonly referred to as the system level parameterization (SLP)~\cite{anderson2019system}. As it is an affine constraint, we can leverage it to derive a convex optimization problem that optimizes over disturbance feedback matrices.

Throughout the paper, we make use of the SLP to formulate SMPC optimization problems. Although applying the SLP in an MPC problem is straightforward, the difficulty comes from ensuring recursive feasibility and closed-loop chance constraint satisfaction despite the presence of unbounded disturbances and online-optimized feedback.

\section{Stochastic MPC with Online-optimized Feedback: The Conceptual Idea} \label{sec:concept}

In this section, we provide a roadmap for the rest of the paper, and propose the conceptual stochastic MPC problem.

We leverage the system level parameterization presented in~\ref{sec:SLP} to jointly optimize over the disturbance feedback matrices and the nominal trajectories inside the SMPC problem. As customary in MPC, the finite horizon problem is repeatedly solved every time step in a receding horizon manner. In order to ensure the desired infinite horizon closed-loop guarantees, we rely on a custom-designed terminal set and appropriate modification of the constraints.

The proposed SMPC optimization is of the form
\begin{subequations} \label{eq:conceptual-SMPC}
    \begin{align}
        \min_{\substack{\mathbf{z}_k, \mathbf{v}_k\\ \mathbf{\Phi}^\mathrm{x}_{:|k}, \mathbf{\Phi}^\mathrm{u}_{:|k}}} \quad &J(\mathbf{z}_k, \mathbf{v}_k, \mathbf{\Phi}^\mathrm{x}_{:|k}, \mathbf{\Phi}^\mathrm{u}_{:|k})  \label{eq:conceptual-SMPC:cost}\\
        \text{s.t.} \quad\ \ 
        &z_{0|k} = x_k, \label{eq:conceptual-SMPC:init}\\
        &z_{i+1|k} = A z_{i|k} + B v_{i|k},\\
        &(z_{i|k}, v_{i|k}, \mathbf{\Phi}^\mathrm{x}_{i|k}, \mathbf{\Phi}^\mathrm{u}_{i|k}) \in \mathcal{C}^j_{i|k}, \label{eq:conceptual-SMPC:constraints_N}\\
        &(z_{N|k}, \mathbf{\Phi}^\mathrm{x}_{N|k}) \in \mathcal{C}_{N|k}, \label{eq:conceptual-SMPC:constraints_terminal}\\
        &\text{the SLP~\eqref{eq:SLP} is satisfied}, \\
        &\forall{j} \in \mathbb{I}_{[1,c]},\ \forall i \in \mathbb{I}_{[0,N-1]}.
    \end{align}
\end{subequations}
For predictions computed at time step $k \in \mathbb{I}_{\geq 0}$, we use the notation $\mathbf{z}_k = z_{0:N|k}$, $\mathbf{v}_k = v_{0:N-1|k}$. Furthermore, define $\mathbf{\Phi}^\mathrm{x}_{:|k}$ and $\mathbf{\Phi}^\mathrm{u}_{:|k}$ as block-lower-triangular matrices, where the $i$-th row of the lower triangular part is $\mathbf{\Phi}^\mathrm{x}_{i|k}$ and $\mathbf{\Phi}^\mathrm{u}_{i|k}$, respectively.
That is, the SMPC~\eqref{eq:conceptual-SMPC} initializes the predicted state trajectory $\mathbf{z}_k$ with the current state $x_k$, and makes use of the SLP presented in Section~\ref{sec:SLP} to propagate uncertainty and optimize over the feedback matrices online. The objective function~\eqref{eq:conceptual-SMPC:cost} is a suitable quadratic function defined later in Section~\ref{sec:FHOCP} (Eq.~\eqref{eq:J}).

In the following sections, we detail how we build up Problem~\eqref{eq:conceptual-SMPC} and the corresponding theoretical properties.
First, Section~\ref{sec:DMC} shows how to formulate a finite horizon optimal control problem that will serve as the initial SMPC problem at time step $k=0$. In order to achieve infinite horizon guarantees in terms of closed-loop chance constraint satisfaction, we design the terminal set in Section~\ref{sec:Mode_2}, which is our first main contribution. Then, when the problem is re-solved at time step $k>0$, the constraints of~\eqref{eq:conceptual-SMPC} have to be adapted to ensure both recursive feasibility and closed-loop chance constraint satisfaction. This is achieved via an intermediate result in Section~\ref{sec:reconditioning}, where we leverage the idea of reconditioning the chance constraints on all information available up to the current time step to design recursively feasible SMPC algorithms. Finally, this allows us to bring everything together and define the receding horizon SMPC algorithm in Section~\ref{sec:RHC}. 

Note that existing SMPC methods with online-optimized feedback (e.g.~\cite{oldewurtel2008tractable, mark2022recursively}) can also be represented in the form of~\eqref{eq:conceptual-SMPC}. The main difference compared to them is the choice of the constraints~\eqref{eq:conceptual-SMPC:constraints_N}-\eqref{eq:conceptual-SMPC:constraints_terminal}. By appropriately defining constraint sets $\mathcal{C}^j_{i|k}$ and $\mathcal{C}_{N|k}$, we can guarantee recursive feasibility of~\eqref{eq:conceptual-SMPC} and satisfaction of the chance constraints~\eqref{eq:chance_constraints} for the resulting closed-loop system.

\section{Finite Horizon Stochastic Optimal Control Using Optimized Feedback}
\label{sec:DMC}

In the following, we present the finite horizon optimal control problem that comes with chance constraint satisfaction guarantees for the infinite horizon. This will serve as the optimization problem at $k=0$ for the final SMPC scheme presented in Section~\ref{sec:RHC}.

The prediction horizon is divided into two parts. In the first part, a finite prediction horizon $N\in\mathbb{I}_{>0}$ is used over which the affine policy is optimized (see Section~\ref{sec:Mode_1}). As for the infinite-horizon tail, a fixed linear feedback is applied, which, using the maximal probabilistic admissible set designed in Section~\ref{sec:Mode_2}, is able to guarantee chance constraint satisfaction for the infinite horizon. The resulting finite horizon optimal control problem is summarized in Section~\ref{sec:FHOCP}.

\subsection{Probabilistic Constraints with Optimized Feedback} \label{sec:Mode_1}

First, we show how we can reformulate the original chance constraints~\eqref{eq:chance_constraints} to optimize over affine feedback policies in the first $N$ steps of the prediction horizon. Consider system~\eqref{eq:system} subject to chance constraints~\eqref{eq:chance_constraints}, and the SLP presented in Section~\ref{sec:SLP}:
\begin{equation} \label{eq:SLP:ux}
    \begin{split}
    u_{i|0} &= v_{i|0} + e^\mathrm{u}_{i|0},\ \forall i \in \mathbb{I}_{[0,N-1]},\\
    x_{i|0} &= z_{i|0} + e^\mathrm{x}_{i|0},\ \forall i \in \mathbb{I}_{[0,N]},
    \end{split}
\end{equation}
where $e^\mathrm{u}_{i|0}$ and $e^\mathrm{x}_{i|0}$ are defined in~\eqref{eq:e^u_k},~\eqref{eq:e^x_k}. Due to the i.i.d. Gaussian assumption on the disturbance sequence, $x_{i|0}$ and $u_{i|0}$ are Gaussian random variables with means
\begin{align*}
    \mathbb{E}[u_{i|0}] &= v_{i|0},\ \forall i \in \mathbb{I}_{[0,N-1]},\\
    \mathbb{E}[x_{i|0}] &= z_{i|0},\ \forall i \in \mathbb{I}_{[0,N]},
\end{align*}
and the variance of the joint distribution is given by
\begin{equation*}
    \mathrm{Var}[[x_{i|0}^\top, u_{i|0}^\top]^\top] = \begin{bmatrix}
        \mathbf{\Phi}^\mathrm{x}_{i|0} \\ \mathbf{\Phi}^\mathrm{u}_{i|0}
    \end{bmatrix} \mathbf{\Sigma}^\mathrm{w}_i \begin{bmatrix}
        (\mathbf{\Phi}^{\mathrm{x}}_{i|0})^\top & (\mathbf{\Phi}^{\mathrm{u}}_{i|0})^\top
    \end{bmatrix},
\end{equation*}
$\forall i \in \mathbb{I}_{[0,N-1]}$, where $\mathbf{\Sigma}^\mathrm{w}_i = I_i \otimes \Sigma^\mathrm{w}$. 
Next, we provide an exact reformulation of the probabilistic half-space constraints~\eqref{eq:chance_constraints} by extending~\cite[Sec.~3.2]{hewing2020recursively} to handle optimized feedback matrices.
\begin{lemma} \label{lem:SOC_reform}
    Consider the parameterization~\eqref{eq:SLP:ux}. Then, the chance constraint
    \begin{equation} \label{eq:chance_constraint_j}
        \mathrm{Pr}[(x_{i|0}, u_{i|0}) \in \mathcal{C}^j] \geq p^j
    \end{equation}
    is equivalent to the second-order cone constraint
    \begin{equation} \label{eq:SOC_k=0_Psi}
    \begin{split}
         G_j z_{i|0} &+ H_j v_{i|0} \leq b_j \\&- \sqrt{\tilde{p}^j} \norm{(\mathbf{\Sigma}^\mathrm{w}_i)^{1/2} \begin{bmatrix} (\mathbf{\Phi}^{\mathrm{x}\top}_{i|0})^\top & (\mathbf{\Phi}^{\mathrm{u}}_{i|0})^\top\end{bmatrix} \begin{bmatrix} G_j^\top \\ H_j^\top \end{bmatrix}}
    \end{split}
    \end{equation}
    which is further equivalent to
    \begin{equation} \label{eq:SOC_k=0}
    \begin{split}
        G_j z_{i|0} &+ H_j v_{i|0} \leq b_j \\ &- \sqrt{\tilde{p}^j} \norm{\left(\begin{bmatrix} G_j & H_j \end{bmatrix} \otimes (\mathbf{\Sigma}^\mathrm{w}_i)^{1/2}\right) \psi_{i|0}}.
    \end{split}
    \end{equation}
     where $\psi_{i|0} := \mathrm{vec}\left(\begin{bmatrix} (\mathbf{\Phi}^{\mathrm{x}}_{i|0})^\top & (\mathbf{\Phi}^{\mathrm{u}}_{i|0})^\top\end{bmatrix}\right)$ and $\tilde{p}^j := \chi^2(2p^j-1)$.
\end{lemma}
\begin{proof}
    The proof is similar to that of~\cite[Sec. 3.2]{hewing2020recursively} and it can be found in Appendix~\ref{app:SOC}.
\end{proof}
\begin{remark}
    With a slight abuse of notation, for $i=0$ the matrices $\mathbf{\Sigma}^\mathrm{w}_0,$ $\mathbf{\Phi}^\mathrm{x}_{0|0}$ and $\mathbf{\Phi}^\mathrm{u}_{0|0}$ are defined as 0, in which case~\eqref{eq:SOC_k=0_Psi} results in an untightened half-space constraint on $z_{0|0}$ and $v_{0|0}$.
\end{remark}

\subsection{Probabilistic Constraints with Terminal Controller: Terminal Set Design} \label{sec:Mode_2}
In this section we consider the infinite tail of the prediction horizon by designing a suitable terminal set such that the desired chance constraint~\eqref{eq:chance_constraints} is satisfied at all times.

In order to keep the number of decision variables finite, a fixed stabilizing linear state feedback $K$ is used in the tail of the horizon:
\begin{equation} \label{eq:Mode_2_terminal_controller}
    u_{N+i|0} = Kx_{N+i|0},\ \forall i \in \mathbb{I}_{\geq 0},
\end{equation}
where $K \in \mathbb{R}^{m\times n}$ is chosen such that $A_K := A + BK$ is Schur stable. Therefore,
\begin{align*}
    \mathbb{E}[u_{N+i|0}] &= KA_K^iz_{N|0},\ \forall i \in \mathbb{I}_{\geq 0},\\
    \mathbb{E}[x_{N+i|0}] &= A_K^iz_{N|0},\ \forall i \in \mathbb{I}_{\geq 0},
\end{align*}
and
\begin{equation*}
\begin{split}
    \mathrm{Var}[[x_{N+i|0}^\top,&\, u_{N+i|0}^\top]^\top] = \\ &\begin{bmatrix}
        I_n \\ K
    \end{bmatrix}\begin{bmatrix}
        A_K^i\mathbf{\Phi}^\mathrm{x}_{N|0} & I_n
    \end{bmatrix} \begin{bmatrix}
        \mathbf{\Sigma}^\mathrm{w}_N & \\ & \Sigma^\mathrm{x}_i
    \end{bmatrix} [*]^\top,
\end{split}
\end{equation*}
where $\Sigma^\mathrm{x}_i$ is defined by the recursion
\begin{equation*}
    \Sigma^\mathrm{x}_{i+1} = A_K\Sigma^\mathrm{x}_iA_K^\top + \Sigma^\mathrm{w},\ \Sigma^\mathrm{x}_0 = 0.
\end{equation*}
Note that since $A_K$ is Schur stable, the
limit $\lim_{i \rightarrow \infty} \Sigma^\mathrm{x}_i = \Sigma^\mathrm{x}_\infty$ exists, and it is given by the Lyapunov equation
\begin{equation*}
    \Sigma^\mathrm{x}_\infty = A_K \Sigma^\mathrm{x}_\infty A_K^\top + \Sigma^\mathrm{w}.
\end{equation*}

\par Following similar arguments to Lemma~\ref{lem:SOC_reform}, the chance constraint
\begin{equation} \label{eq:chance_constraint_tail}
    \mathrm{Pr}[(x_{N+i|0}, u_{N+i|0}) \in \mathcal{C}^j] \geq p^j,
\end{equation}
is equivalent to the second-order cone constraint
\begin{equation} \label{eq:SOC_k=0_tail_psi}
    G_{K,j} A_K^i z \leq b_j - \sqrt{\tilde{p}^j} \norm{\begin{bmatrix}
        \left( G_{K,j} A_K^i \otimes (\mathbf{\Sigma}^{\mathrm{w}}_N)^{1/2}\right) \psi\\ (\Sigma_i^{\mathrm{x}})^{1/2} G_{K,j}^\top
    \end{bmatrix}},
\end{equation}
where $G_{K,j}$ is the $j$-th row of $G_K = G + HK \in \mathbb{R}^{c \times n}$, $z := z_{N|0}$ is used for notational convenience, and $\psi := \mathrm{vec}(\mathbf{\Phi}^{\mathrm{x}\top}_{N|0})$. As~\eqref{eq:SOC_k=0_tail_psi} needs to hold for all $i \in \mathbb{I}_{\geq 0}$, the number of constraints is infinite. We address this issue by proving that imposing a finite subset of the constraints is equivalent to imposing all infinite constraints.

The following result builds on the well-established field of maximal admissible set theory~\cite{gilbert1991linear}. The novelty comes from the fact that, $(i)$ unlike existing results that construct maximal admissible sets for a deterministic state subject to deterministic constraints, we have to construct such a set for a random variable subject to probabilistic constraints; $(ii)$ it results in a maximal admissible set that is not simply a set for the state $z$, but also for the feedback matrix $\psi$; and $(iii)$ in our case the set must satisfy an infinite number of \emph{time-varying} constraints (due to the presence of $\Sigma^\mathrm{x}_i$ in~\eqref{eq:SOC_k=0_tail_psi}). Compared to existing SMPC approaches that use online-optimized feedback~\cite{oldewurtel2008tractable, prandini2012randomized, mark2022recursively, pan2023data, li2024distributionally, knaup2024recursively}, we provide stronger guarantees in terms of closed-loop chance constraint satisfaction through this terminal set design.
\par
Consider the tightened version of inequality~\eqref{eq:SOC_k=0_tail_psi}:
\begin{equation}
     G_{K,j} A_K^i z \leq b_j - \sqrt{\tilde{p}^j} \norm{\begin{bmatrix}
        \left( G_{K,j} A_K^i \otimes (\mathbf{\Sigma}^{\mathrm{w}}_N)^{1/2}\right) \psi\\ (\Sigma_\infty^{\mathrm{x}})^{1/2} G_{K,j}^\top
    \end{bmatrix}}, \label{eq:SOC_k=0_tail_psi_tightened}
\end{equation}
where $\Sigma^\mathrm{x}_i$ has been replaced by $\Sigma^\mathrm{x}_\infty \succeq \Sigma^\mathrm{x}_i$. Let
\begin{align*}
    \mathcal{D}_i = \myset{(z, \psi)}{\eqref{eq:SOC_k=0_tail_psi} \text{ holds for }i,\ \forall j \in \mathbb{I}_{[1, c]}}, \\
    \bar{\mathcal{D}}_i = \myset{(z, \psi)}{\eqref{eq:SOC_k=0_tail_psi_tightened} \text{ holds for }i,\ \forall j \in \mathbb{I}_{[1, c]}},
\end{align*}
and define the $i$-step maximal admissible set for the probabilistic constraints~\eqref{eq:chance_constraint_tail} via the recursion
\begin{equation} \label{eq:S_sequence}
    \mathcal{S}_i = \mathcal{S}_{i-1} \cap \mathcal{D}_i,\quad \mathcal{S}_0 = \mathcal{D}_0.
\end{equation}
Finally, the maximal admissible set is defined as
\begin{equation*}
    \mathcal{S}_\infty = \myset{(z, \psi)}{\eqref{eq:SOC_k=0_tail_psi} \text{ holds}\ \forall i \in \mathbb{I}_{\geq 0}, \forall j \in \mathbb{I}_{[1, c]}}.
\end{equation*}
Note that sequence~\eqref{eq:S_sequence} converges to $\mathcal{S}_\infty$ and 
$\mathcal{S}_\infty \subseteq \mathcal{S}_i \subseteq \mathcal{S}_{i-1}$. Using $\bar{\mathcal{D}}_i$, sets $\bar{\mathcal{S}}_i$ and $\bar{\mathcal{S}}_\infty$ are defined analogously.
\begin{assumption} \label{ass:L}
    Consider the decomposition $G = LC$ and $H = LD$, where $L \in \mathbb{R}^{c\times d}$, $C \in \mathbb{R}^{d\times n}$, $D \in \mathbb{R}^{d\times m}$. We assume that the set $\mathcal{L} = \myset{y \in \mathbb{R}^d}{Ly\leq b}$ is bounded, and the pair $(C_K, A_K)$ is observable, where $A_K = A + BK$ and $C_K = C + DK$ for the stabilizing terminal feedback gain $K$. Finally,
    \begin{equation} \label{eq:b_j-norm>0}
        b_j -  \sqrt{\tilde{p}^j}\norm{(\Sigma_\infty^{\mathrm{x}})^{1/2} C_K^\top L_j^\top} > 0,\ \forall j \in \mathbb{I}_{[1,c]}.
    \end{equation}
    where $L_j$ denotes the $j$-th row of $L$.
\end{assumption}
\begin{remark} \label{rem:L_decomposition}
    The decomposition in Assumption~\ref{ass:L} naturally occurs whenever the system is subject to polytopic constraints with respect to the lower dimensional variable $y = Cx + Du$. Furthermore, note that $(x,Kx) \in \cap_{j=1}^c \mathcal{C}^j \iff C_K x \in \mathcal{L}$. The boundedness of $\mathcal{L}$ is a standard technical assumption in the maximal admissible set literature~\cite{gilbert1991linear}.
\end{remark}
\begin{remark}
    If~\eqref{eq:b_j-norm>0} does not hold, then applying the linear feedback $K$ does not ensure constraint satisfaction, even when initialized at $x=0$. Note that a similar condition is required by any SMPC approach using a fixed feedback, e.g.,~\cite{hewing2018stochastic, hewing2020recursively, schluter2022stochastic, kohler2022recursively}. A convex optimization problem that produces a suitable $K$ is presented in Appendix~\ref{app:K_design}.
\end{remark}

\begin{theorem}[Finitely determined terminal set] \label{thm:terminal_set}
    Suppose Assumption~\ref{ass:L} holds. Then, the maximal admissible set $\mathcal{S}_\infty$ is finitely determined, i.e., there exists $\mu \in \mathbb{I}_{\geq 0}$ such that $\mathcal{S}_\infty = \mathcal{S}_\mu$. Furthermore, the finite determination index $\mu$ is the smallest non-negative integer that satisfies
    \begin{equation} \label{eq:terminal_set_Smu_inclusion}
        \mathcal{S}_\mu \subseteq \bigcap_{i=\mu+1}^\infty \bar{\mathcal{D}}_i.
    \end{equation}
\end{theorem}
\begin{proof}
    The proof can be found in Appendix~\ref{app:terminal_set_proof}.
\end{proof}

While Theorem~\ref{thm:terminal_set} ensures the existence of index $\mu$, its computation involves checking the set inclusion~\eqref{eq:terminal_set_Smu_inclusion}, which is non-trivial, as the right hand side is an intersection of an infinite number of sets. In order to address this, we propose Algorithm~\ref{alg:terminal_set} for finding $\mu$, leveraging results on maximal admissible sets~\cite[Sec.~III--IV]{gilbert1991linear}.

\begin{algorithm}
\caption{Terminal set design} \label{alg:terminal_set}
\begin{algorithmic}[1]
    \State Increase $\nu \in \mathbb{I}_{\geq 0}$ until $\bar{\mathcal{S}}_\nu \subseteq \bar{\mathcal{D}}_{\nu+1}$. \label{algstep:nu}
    \State Increase $\mu \in \mathbb{I}_{\geq 0}$ until $\mathcal{S}_\mu \subseteq \cap_{i=\mu+1}^{\nu+\mu+1}\bar{\mathcal{D}}_i$. \label{algstep:mu}
    \State Construct $\mathcal{S}_\mu$ according to~\eqref{eq:S_sequence}.
    \Ensure $\mathcal{S}_\mu$
\end{algorithmic}
\end{algorithm}

Algorithm~\ref{alg:terminal_set} takes advantage of the insight that there exists an index $\nu > 0$ such that $\cap_{i=\mu+1}^{\nu+\mu+1}\bar{\mathcal{D}}_i = \cap_{i=\mu+1}^\infty\bar{\mathcal{D}}_i$. Step~\ref{algstep:nu} finds the smallest such integer, and Step~\ref{algstep:mu} then checks the containment of $\mathcal{S}_\mu$ in $\cap_{i=\mu+1}^{\nu+\mu+1}\bar{\mathcal{D}}_i$. We formally prove this insight in Appendix~\ref{app:algorithm_proof}.

\begin{proposition} \label{prop:terminal_set_alg}
    Suppose Assumption~\ref{ass:L} holds. Then, Algorithm~\ref{alg:terminal_set} is guaranteed to terminate and the set $\mathcal{S}_\mu$ that it returns satisfies~\eqref{eq:terminal_set_Smu_inclusion}, and hence $\mathcal{S}_\mu = \mathcal{S}_\infty$.
\end{proposition}
\begin{proof}
    The proof can be found in Appendix~\ref{app:algorithm_proof}.
\end{proof}

\begin{remark} \label{rem:lossy_S_procedure}
    While Algorithm~\ref{alg:terminal_set} addresses the issue of computing the intersection of infinite number of sets in~\eqref{eq:terminal_set_Smu_inclusion}, the set inclusions in Steps~\ref{algstep:nu}--\ref{algstep:mu} are still challenging, as they require testing containment of second-order cones. We propose sufficient conditions utilizing the `lossy' S-procedure~\cite{boyd1994linear} in Appendices~\ref{app:finding_nu}--\ref{app:finding_mu}. However, if these sufficient conditions are used, then Algorithm~\ref{alg:terminal_set} may not terminate, which is a general issue when constructing maximal admissible sets for nonlinear constraints~\cite[Section~III]{gilbert1991linear}.
\end{remark}
\begin{remark} \label{rem:can_still_converge}
    Assumption~\ref{ass:L} is a sufficient but not necessary condition for the existence of a finite $\mu$, i.e., it is still possible that $\mathcal{S}_\infty$ is finitely determined, even when Assumption~\ref{ass:L} does not hold (e.g. when $\mathcal{L}$ is not bounded). In particular, if Algorithm~\ref{alg:terminal_set} converges, then $\mathcal{S}_\infty = \mathcal{S}_\mu$, i.e., the terminal set design is still successful, regardless of Assumption~\ref{ass:L}.
\end{remark}

\subsection{Proposed Optimization Problem} \label{sec:FHOCP}

Sections~\ref{sec:Mode_1} and~\ref{sec:Mode_2} presented an exact reformulation of the chance constraint~\eqref{eq:chance_constraints} in the first $N$ steps using optimized affine feedback, and in the infinite horizon tail given the fixed linear controller~\eqref{eq:Mode_2_terminal_controller}, respectively. Using these results, this section presents the proposed finite horizon optimal control problem for time step $k=0$ in the form of the following convex second-order cone program (SOCP):

\begin{subequations} \label{eq:DMC}
    \begin{align}
        \min_{\substack{\mathbf{z}_0, \mathbf{v}_0\\ \mathbf{\Phi}^\mathrm{x}_{:|0}, \mathbf{\Phi}^u_{:|0}}} \quad &J(\mathbf{z}_0, \mathbf{v}_0, \mathbf{\Phi}^\mathrm{x}_{:|0}, \mathbf{\Phi}^\mathrm{u}_{:|0})\\
        \text{s.t.} \quad\ \ 
        &z_{0|0} = x_0, \\
        &z_{i+1|0} = A z_{i|0} + B v_{i|0},\\
        &(z_{i|0}, v_{i|0}, \mathbf{\Phi}^\mathrm{x}_{i|0}, \mathbf{\Phi}^\mathrm{u}_{i|0}) \in \mathcal{C}_{i|0}, \label{eq:DMC:constraints_N}\\
        &(z_{N|0}, \mathrm{vec}(\mathbf{\Phi}^\mathrm{x}_{N|0})) \in \mathcal{S}_\mu, \label{eq:DMC:constraints_terminal}\\
        &\text{the SLP~\eqref{eq:SLP} is satisfied},\\
        &\forall i \in \mathbb{I}_{[0,N-1]},
    \end{align}
\end{subequations}
where
\begin{equation*}
    \mathcal{C}_{i|0} = \myset{(z_{i|0}, v_{i|0}, \mathbf{\Phi}^\mathrm{x}_{i|0}, \mathbf{\Phi}^\mathrm{u}_{i|0})}{\eqref{eq:SOC_k=0_Psi}\text{ is satisfied }\forall j \in \mathbb{I}_{[1,c]}},
\end{equation*}
and the cost function to be minimized is defined as
\begin{align} 
        J(\mathbf{z}_k,&\, \mathbf{v}_k, \mathbf{\Phi}^\mathrm{x}_{:|k}, \mathbf{\Phi}^\mathrm{u}_{:|k}) = \sum_{i=0}^{N-1} \ell(z_{i|k},v_{i|k}) \nonumber\\
        &+ \sum_{i=1}^{N-1}(\mathrm{tr}(Q\mathbf{\Phi}^\mathrm{x}_{i|k}\mathbf{\Sigma}^\mathrm{w}_i(\mathbf{\Phi}^{\mathrm{x}}_{i|k})^\top) + \mathrm{tr}(R\mathbf{\Phi}^\mathrm{u}_{i|k}\mathbf{\Sigma}^\mathrm{w}_i(\mathbf{\Phi}^{\mathrm{u}}_{i|k})^\top)) \nonumber\\ 
        &+ \ell_\mathrm{f}(z_{N|k}) + \mathrm{tr}(P\mathbf{\Phi}^\mathrm{x}_{N|k}\mathbf{\Sigma}^\mathrm{w}_N(\mathbf{\Phi}^{\mathrm{x}}_{N|k})^\top). \label{eq:J}
\end{align}
The terminal cost is defined as $\ell_\mathrm{f}(z) = \| z\|_P^2 +p_\mathrm{f}^\top z$ (see also~\cite{amrit2011economic}), with
\begin{equation} \label{eq:p_Lyapunov}
    p_\mathrm{f} = (I_n - A_K^\top)^{-1}(K^\top r + q),
\end{equation}
which is well-defined since $A_K$ is Schur, and $P$ is the unique positive definite solution of the Lyapunov equation
    \begin{equation} \label{eq:P_Lyapunov}
        A_K^\top P A_K + Q + K^\top R K = P.
    \end{equation}
The number of decision variables of~\eqref{eq:DMC} is $\mathcal{O}(N^2n(n+m))$, and the number of constraints is $\mathcal{O}(N(c+Nn^2) +c\mu)$.
    
\par After solving optimization problem~\eqref{eq:DMC} once at time step $k=0$, we apply
\begin{equation} \label{eq:dual_mode_law}
\begin{split}
    u_k &= v^\star_{k|0} + \sum_{j=1}^k \Phi^{\mathrm{u}\star}_{k,j|0} w_{j-1},\ \forall k \in \mathbb{I}_{[0, N-1]}, \\
    u_k &= K x_k,\ \forall k \in \mathbb{I}_{\geq N}.
\end{split}
\end{equation}
\begin{proposition}[Chance constraint satisfaction and bounded asymptotic average cost.] \label{prop:DMC}
    Suppose that Assumption~\ref{ass:L} holds and optimization problem~\eqref{eq:DMC} is feasible. Then, system~\eqref{eq:system} controlled by policy~\eqref{eq:dual_mode_law} satisfies chance constraints~\eqref{eq:chance_constraints}. Furthermore, the asymptotic average cost is upper bounded by the performance of the static state feedback $K$:
    \begin{equation} \label{eq:lavg}
        \limsup_{T \rightarrow \infty} \frac{1}{T}\sum_{k=0}^{T-1}\mathbb{E}[\ell(x_k,u_k)] \leq \mathrm{tr}(P\Sigma^\mathrm{w}).
    \end{equation}
\end{proposition}
\begin{proof}
    \textit{Part I: Chance constraint satisfaction.} Constraint~\eqref{eq:DMC:constraints_N} corresponds to~\eqref{eq:SOC_k=0_Psi}, which is equivalent to the original chance constraint~\eqref{eq:chance_constraints} for $k \in \mathbb{I}_{[0,N-1]}$. Furthermore, the terminal constraint~\eqref{eq:DMC:constraints_terminal} ensures that the terminal controller satisfies~\eqref{eq:chance_constraints} $k \in \mathbb{I}_{\geq N}$, since it uses the terminal set $\mathcal{S}_\infty$ satisfying chance constraints~\eqref{eq:chance_constraints} by definition, due to the equivalence of~\eqref{eq:chance_constraint_tail} and~\eqref{eq:SOC_k=0_tail_psi}.
\par \textit{Part II: Bounded asymptotic average cost}. 
The expected cost incurred in the first $N$ steps is a finite number, and thus it does not contribute to the asymptotic average cost in the infinite limit. Furthermore, since $P$ satisfies the Lyapunov equation~\eqref{eq:P_Lyapunov}, the asymptotic average cost associated with the infinite tail can be shown to be bounded by $\mathrm{tr}(P\Sigma^\mathrm{w})$ (cf.~\cite[Chapter~8]{kouvaritakis2016model}).
\end{proof}

Proposition~\eqref{prop:DMC} ensures that the affine policy~\eqref{eq:dual_mode_law} resulting from Problem~\eqref{eq:DMC} provides the desired guarantees of satisfaction of probabilistic constraints and a suitable bound on the expected cost. However, this result only applies to the policy optimized at time $k=0$, when only the first $N$ control inputs are optimized. Furthermore, the control policy is not updated based on the information available up to the current time step (i.e., based on $w_{0:k-1}$).

One existing method to re-optimize in closed loop is the indirect feedback paradigm~\cite{hewing2020recursively}.
In order to extend this method to use online-optimized feedback, a natural idea is to combine it with~\eqref{eq:DMC}. However, the crucial assumption that the error system is independent from the optimization problem of the indirect feedback SMPC breaks down in the optimized feedback setting, which leads to the loss of closed-loop guarantees. We illustrate this in Section~\ref{sec:numerical} on a numerical example.

In order to overcome this challenge, we propose a novel SMPC framework in the following sections that relies on the idea of reconditioning.

\section{General Reconditioning Framework for SMPC} \label{sec:reconditioning}

In the following, we present a general receding horizon reconditioning SMPC framework, which is then used in Section~\ref{sec:RHC} to propose the receding horizon implementation of the finite horizon optimization problem derived in Section~\ref{sec:DMC}. This method is inspired by~\cite{wang2021recursive}, where a shrinking horizon reconditioning SMPC is applied for the case of mission-wide (joint-in-time) state constraints. The following method extends this idea to address general causal policies depending on past disturbances and a general receding horizon setting.
\par Consider system~\eqref{eq:system}. The goal is to ensure satisfaction of~\eqref{eq:chance_constraints} for optimized feedback, which we achieve using the idea of reconditioning. To this end, we propose the following optimization problem to be solved at each time step $k \in \mathbb{I}_{\geq 0}$:
\begin{subequations} \label{eq:Rec-SMPC}
    \begin{align}
        \min_{\boldsymbol{\pi_k}} \quad &\mathcal{J}(x_k, \boldsymbol{\pi_k})\\
        \text{s.t.} \quad\, &u_t = \pi_{t-k|k}(w_{k:t-1}),\\
        &\mathrm{Pr}[(x_t, u_t) \in \mathcal{C}^j | w_{0:k-1}, \{\boldsymbol{\pi}^\mathrm{p}_k,\boldsymbol{\pi}_k\}] \geq p^j_{t-k|k},\label{eq:Rec-SMPC:constraints} \\
        &\eqref{eq:system},\ \forall j \in \mathbb{I}_{[1,c]},\ \forall t \in \mathbb{I}_{\geq k},
    \end{align}
\end{subequations}
where $\mathcal{J}$ is some suitable cost function. The optimization is carried out over a sequence of causal policies $\boldsymbol{\pi}_k = \{\pi_{t-k|k}\}_{t=k}^\infty$, where each policy $\pi_{t-k|k}$ maps the disturbances $w_{k:t-1}$ to a control input $u_t$ for $t \in \mathbb{I}_{>k}$, and $\pi_{0|k} \in \mathbb{R}^m$.
The notation $\{\boldsymbol{\pi}^\mathrm{p}_k,\boldsymbol{\pi}_k\}$ denotes the concatenation of the past optimal policies $\boldsymbol{\pi}^\mathrm{p}_k = \{\pi_{0|l}^\star\}_{l=0}^{k-1}$ and the currently optimized future policies $\boldsymbol{\pi}_k$. Furthermore, note that $(x_t, u_t)$ are random variables that depend on every disturbance $w_{0:t-1}$ and also on the sequence of policies. After conditioning on $w_{0:k-1}$ and fixing the policy sequence to $\{\boldsymbol{\pi}^\mathrm{p}_k,\boldsymbol{\pi}_k\}$, $(x_t, u_t)$ are still random variables, however, the stochasticity now only comes from the disturbances that arise between the time step $k$ when the optimization problem is solved and the predicted time step $t$, i.e., $w_{k:t-1}$. Therefore, $\mathrm{Pr}[(x_t, u_t) \in \mathcal{C}^j | w_{0:k-1}, \{\boldsymbol{\pi}^\mathrm{p}_k,\boldsymbol{\pi}_k\}]$ represents the conditional probability that $(x_t, u_t) \in \mathcal{C}^j$, conditioned on the realized values of the past disturbances $w_{0:k-1}$ and fixing the policy sequence to $\{\boldsymbol{\pi}^\mathrm{p}_k,\boldsymbol{\pi}_k\}$.
This corresponds to the predicted future state and input trajectories based on the current measured state, which is uniquely invoked by the past policies $\boldsymbol{\pi}^\mathrm{p}_k$ and the past disturbances.

The required probability levels are updated according to
\begin{equation} \label{eq:initial_prob}
    p^j_{t|0} = p^j
\end{equation}
$\forall j \in \mathbb{I}_{[1,c]},\ \forall t \in \mathbb{I}_{\geq 0}$ and
\begin{equation} \label{eq:pjik}
    p^j_{t-k|k} = \mathrm{Pr}[(x_t, u_t) \in \mathcal{C}^j | w_{0:k-1}, \{\boldsymbol{\pi}^\mathrm{p}_{k-1},\boldsymbol{\pi}_{k-1}^\star\}], \ \forall k \in \mathbb{I}_{> 0},
\end{equation}
$\forall j \in \mathbb{I}_{[1,c]},\ \forall t \in \mathbb{I}_{\geq k}$.
This implies that the new policy $\boldsymbol{\pi}_k$ must satisfy the constraints with a probability that is not lower than that of the previous policy $\boldsymbol{\pi}^\star_{k-1}$, conditioned on all realized disturbances $w_{0:k-1}$.
Note that $\boldsymbol{\pi}_{k-1}^\star$ was computed when $w_{k-1}$ has not been realized yet, so one can think of~\eqref{eq:pjik} as the constraint satisfaction probability achieved by the previous policy sequence in hindsight, knowing the value of the disturbance $w_{k-1}$.

After solving the optimization problem, the first element of the optimal input sequence is applied to the system:
\begin{equation} \label{eq:Rec-SMPC:control_law}
    u_k = \pi^\star_{0|k}.
\end{equation}
\begin{proposition}[Recursive feasibility and closed-loop chance constraint satisfaction of reconditioning SMPC] \label{prop:Rec-SMPC}
    Assume that~\eqref{eq:Rec-SMPC} is feasible at time step $k=0$. Then, problem~\eqref{eq:Rec-SMPC} remains feasible $\forall k \in \mathbb{I}_{\geq 0}$ for the closed-loop system resulting from applying feedback~\eqref{eq:Rec-SMPC:control_law} to system~\eqref{eq:system}, and the chance constraints~\eqref{eq:chance_constraints} are satisfied.
\end{proposition}
\begin{proof}
    The proof builds on the ideas of~\cite[Prop.~1]{wang2021recursive} and it can be found in Appendix~\ref{app:reconditioning_proof}.
\end{proof}

\begin{remark} \label{rem:vanishing_constraints}
    Note that~\eqref{eq:pjik} yields $p^j_{0|k} = 0$ when the optimal input from the previous solution would result in input constraint violation. In this case, condition~\eqref{eq:Rec-SMPC:constraints} allows for unbounded inputs without compromising the probabilistic constraints~\eqref{eq:chance_constraints}.
\end{remark}

\begin{remark} \label{rem:reconditioning_nonlinear}
    Although we consider linear time-invariant systems~\eqref{eq:system} with half-space chance constraints~\eqref{eq:chance_constraints}, the results of this section directly apply to general nonlinear systems affected by i.i.d. disturbances subject to chance constraints, without any modifications.
\end{remark}

\section{Proposed Receding Horizon Reconditioning SMPC Algorithm} \label{sec:RHC}

In this section, we build on the results of Sections~\ref{sec:DMC} and~\ref{sec:reconditioning} to formulate an SMPC method that comes with the desired closed-loop guarantees in terms of constraint satisfaction and performance. As standard in MPC, we re-optimize the policy in a receding horizon fashion to improve performance by leveraging the most recent state measurement and also reduce the effects of using a finite horizon.
The proposed SMPC scheme is presented in Section~\ref{sec:RHC_SMPC} and the closed-loop analysis is carried out in Section~\ref{sec:RHC_analysis}.

\subsection{Proposed SMPC Scheme} \label{sec:RHC_SMPC}
In the following, a tractable receding horizon SMPC scheme is proposed that builds on the idea of reconditioning to re-optimize the feedback policies while ensuring recursive feasibility and closed-loop chance constraint satisfaction. We combine the reconditioning idea from Section~\ref{sec:reconditioning} with the tractable SOCP formulation for optimized SLP-based feedback from Section~\ref{sec:DMC}. In particular, we consider~\eqref{eq:DMC} at $k=0$, and in the following it is explained how reconditioning can be used for $k \in \mathbb{I}_{>0}$.
\par Using the notation from Section~\ref{sec:reconditioning}, let $(x_{i|k},u_{i|k})$ denote $(x_{k+i}, u_{k+i})| w_{0:k-1}, \{\boldsymbol{\pi}^\mathrm{p}_k,\boldsymbol{\pi}_k\}$, i.e., the prediction of the state and input at time step $k+i$, conditioned on the past disturbances $w_{0:k-1}$, keeping the future disturbances $w_{k:k+i-1}$ as random variables, and fixing the policy sequence to $\{\boldsymbol{\pi}^\mathrm{p}_k,\boldsymbol{\pi}_k\}$. For the first $N$ steps, the SLP-based feedback parameterization results in
\begin{equation} \label{eq:RHC:SLS-based_param}
    \begin{split}
        u_{i|k} &= v_{i|k} + \sum_{j=1}^i \Phi^\mathrm{u}_{i,j|k} w_{k+j-1},\ \forall i \in \mathbb{I}_{[0,N-1]}, \\
        x_{i|k} &= z_{i|k} + \sum_{j=1}^i \Phi^\mathrm{x}_{i,j|k} w_{k+j-1},\ \forall i \in \mathbb{I}_{[0,N]}.
    \end{split}
\end{equation}
When conditioned on $w_{0:k-1}$, the system response matrices in~\eqref{eq:RHC:SLS-based_param} are deterministic, and thus $p((x_t, u_t) \in \mathcal{C}^j | w_{0:k-1}, \{\boldsymbol{\pi}^\mathrm{p}_k,\boldsymbol{\pi}_k\})$ and $p((x_t, u_t) \in \mathcal{C}^j | w_{0:k-1}, \{\boldsymbol{\pi}^\mathrm{p}_{k-1},\boldsymbol{\pi}_{k-1}^\star\})$ in~\eqref{eq:Rec-SMPC:constraints} and~\eqref{eq:pjik} are Gaussian. Thus, we can construct a deterministic reformulation of the probabilistic constraint~\eqref{eq:Rec-SMPC:constraints}, similarly to Section~\ref{sec:DMC}. We will use the following notation for the shifted policy from time step $k-1$ reconditioned on $w_{k-1}$:
\begin{equation} \label{eq:hat_zvPsi_def}
    \begin{split}
        \hat{z}_{i|k} &:= z^\star_{i+1|k-1} + \Phi^{\mathrm{x}\star}_{i+1,1|k-1} w_{k-1},\\
    \hat{v}_{i|k} &:= v^\star_{i+1|k-1} + \Phi^{\mathrm{u}\star}_{i+1,1|k-1} w_{k-1},\\
     \hat{\mathbf{\Phi}}^\mathrm{x}_{i|k} &:= \begin{bmatrix} \Phi^{\mathrm{x}\star}_{i+1,2|k-1} & ... & \Phi^{\mathrm{x}\star}_{i+1,i+1|k-1}\end{bmatrix},\\
    \hat{\mathbf{\Phi}}^\mathrm{u}_{i|k} &:= \begin{bmatrix} \Phi^{\mathrm{u}\star}_{i+1,2|k-1} & ... & \Phi^{\mathrm{u}\star}_{i+1,i+1|k-1}\end{bmatrix},
    \end{split}
\end{equation}
$\forall i \in \mathbb{I}_{[0,N-1]}$, where
\begin{equation}
    \begin{split}
         v^\star_{N|k-1} &:= Kz^\star_{N|k-1}, \\
    \Phi^{\mathrm{u}\star}_{N,j|k-1} &:= K\Phi^{\mathrm{x}\star}_{N,j|k-1},\ \forall j \in \mathbb{I}_{[1,N]};
    \end{split}
\end{equation}
furthermore,
\begin{equation} \label{eq:hat_zNPsiN_def}
    \begin{split}
        \hat{z}_{N|k} &:= A_Kz^\star_{N|k-1} + A_K\Phi^{\mathrm{x}\star}_{N,1|k-1}w_{k-1}, \\
    \hat{\mathbf{\Phi}}^\mathrm{x}_{N|k} &:= \begin{bmatrix} A_K\hat{\mathbf{\Phi}}^\mathrm{x}_{N-1|k} & I_n\end{bmatrix}.
    \end{split}
\end{equation}
\par Using this notation, we now show how to define the constraint sets $\mathcal{C}^j_{i|k}$ and $\mathcal{C}_{N|k}$ of the conceptual SMPC problem~\eqref{eq:conceptual-SMPC} in Section~\ref{sec:concept} by leveraging the results on reconditioning from Section~\ref{sec:reconditioning}. In particular, we show how to apply the reconditioning constraints~\eqref{eq:Rec-SMPC:constraints} to the SLP-based parameterization~\eqref{eq:RHC:SLS-based_param}. First, consider $i \in \mathbb{I}_{[0:N-1]}$. It is possible that the conditional distribution in~\eqref{eq:pjik} is degenerate (has zero variance), which occurs exactly when 
\begin{equation} \label{eq:pjik_degenerate}
    G_j\hat{\mathbf{\Phi}}^\mathrm{x}_{i|k} + H_j\hat{\mathbf{\Phi}}^\mathrm{u}_{i|k} = 0.
\end{equation}
In this case, $p^j_{i|k} = 1$ whenever
\begin{equation} \label{eq:previous_satisfied}
    G_j\hat{z}_{i|k} +H_j\hat{v}_{i|k} \leq b_j,
\end{equation}
and $p^j_{i|k} = 0$ otherwise. In particular, in the special case of $i=0$, the variance of $p(x^\star_{1|k-1}, u^\star_{1|k-1}|w_{0:k-1})$ is always 0 (cf. Remark~\ref{rem:vanishing_constraints}). By using the convention that matrices $\hat{\mathbf{\Phi}}^\mathrm{x}_{0|k}$ and $\hat{\mathbf{\Phi}}^\mathrm{u}_{0|k}$ always satisfy~\eqref{eq:pjik_degenerate}, this case is also covered by the general rule discussed above. To handle the degenerate cases, we define the following constraint sets:
\begin{align} 
    \mathcal{C}^{j,(1)}_{i|k} &= \myset{(z_{i|k}, v_{i|k}, \mathbf{\Phi}^\mathrm{x}_{i|k}, \mathbf{\Phi}^\mathrm{u}_{i|k})}{\begin{aligned}
            G_j z_{i|k} + H_j v_{i|k} \leq b_j \\
            G_j\mathbf{\Phi}^\mathrm{x}_{i|k} + H_j\mathbf{\Phi}^\mathrm{u}_{i|k} = 0
         \end{aligned}}, \label{eq:Cijk1} \\
    \mathcal{C}^{j,(2)}_{i|k} &= \mathbb{R}^n \times \mathbb{R}^m \times \mathbb{R}^{n \times in} \times \mathbb{R}^{m \times in}. \nonumber
\end{align}
As for the non-degenerate case (i.e., when~\eqref{eq:pjik_degenerate} does not hold), following similar derivations as in Section~\ref{sec:DMC}, we arrive at inequality
\begin{equation} \label{eq:RHC_SOC_k>0}
\begin{split}
    G_j z_{i|k} &+ H_j v_{i|k} \leq b_j \\ &- \alpha^j_{i|k} \norm{(\mathbf{\Sigma}^{\mathrm{w}}_i)^{1/2} \begin{bmatrix} (\mathbf{\Phi}^{x}_{i|k})^\top & (\mathbf{\Phi}^{u}_{i|k})^\top \end{bmatrix} \begin{bmatrix} G_j^\top \\ H_j^\top \end{bmatrix}},
\end{split}
\end{equation}
where
\begin{equation} \label{eq:alpha_ijk}
    \alpha^j_{i|k} := \frac{b_j-G_j \hat{z}_{i|k}-H_j \hat{v}_{i|k}}{\norm{(\mathbf{\Sigma}^{\mathrm{w}}_i)^{1/2} \begin{bmatrix} (\hat{\mathbf{\Phi}}^{\mathrm{x}}_{i|k})^\top & (\hat{\mathbf{\Phi}}^{\mathrm{u}}_{i|k})^\top\end{bmatrix} \begin{bmatrix} G_j^\top \\ H_j^\top \end{bmatrix}}}.
\end{equation}
Compared to~\eqref{eq:SOC_k=0_Psi} derived for $k=0$, $\alpha^j_{i|k}$ takes the place of $\sqrt{\tilde{p}^j}$. However, unlike $\sqrt{\tilde{p}^j}$, it is possible that $\alpha^j_{i|k}$ is negative, in which case~\eqref{eq:RHC_SOC_k>0} is non-convex. Thus, whenever $\alpha^j_{i|k} < 0$, we enforce the following convex sufficient condition instead:
\begin{equation} \label{eq:suff_convex_neg_alpha}
\begin{split}
     G_j z_{i|k} + H_j v_{i|k} &\leq G_j \hat{z}_{i|k} + H_j \hat{v}_{i|k}, \\
     G_j\mathbf{\Phi}^\mathrm{x}_{i|k} + H_j\mathbf{\Phi}^\mathrm{u}_{i|k} &= G_j\hat{\mathbf{\Phi}}^\mathrm{x}_{i|k} + H_j\hat{\mathbf{\Phi}}^\mathrm{u}_{i|k},
\end{split}
\end{equation}
which simplifies~\eqref{eq:RHC_SOC_k>0} to an inequality on the means by enforcing that the variances are matched. Using this, we define the following constraint sets:
\begin{align*}
    \mathcal{C}^{j,(3)}_{i|k} &= \myset{(z_{i|k}, v_{i|k}, \mathbf{\Phi}^\mathrm{x}_{i|k}, \mathbf{\Phi}^\mathrm{u}_{i|k})}{\eqref{eq:RHC_SOC_k>0} \text{ is satisfied}}, \\
    \mathcal{C}^{j,(4)}_{i|k} &= \myset{(z_{i|k}, v_{i|k}, \mathbf{\Phi}^\mathrm{x}_{i|k}, \mathbf{\Phi}^\mathrm{u}_{i|k})}{\eqref{eq:suff_convex_neg_alpha} \text{ is satisfied}}.
\end{align*}

\par Following a similar derivation for the tail of the infinite horizon, we obtain the constraints
\begin{align} \label{eq:SOC_k>0_tail}
        G_{K,j} A_K^i &z_{N|k} \leq b_j \\ &- \alpha^j_{N+i|k} \norm{\begin{bmatrix}
        (\mathbf{\Sigma}^{\mathrm{w}}_N)^{1/2} (\mathbf{\Phi}^{\mathrm{x}}_{N|k})^\top (A_K^{i})^\top G_{K,j}^\top \\ (\Sigma^{\mathrm{x}}_i)^{1/2} G_{K,j}^\top \nonumber
    \end{bmatrix}}
\end{align}
$\forall i \in \mathbb{I}_{\geq 0}$, where
\begin{equation*}
    \alpha^j_{N+i|k} = \frac{b_j-G_{K,j}A_K^i\hat{z}_{N|k}}{\norm{\begin{bmatrix}
        (\mathbf{\Sigma}^{\mathrm{w}}_N)^{1/2} (\hat{\mathbf{\Phi}}^{\mathrm{x}}_{N|k})^\top (A_K^{i})^\top G_{K,j}^\top \\ (\Sigma^{\mathrm{x}}_i)^{1/2} G_{K,j}^\top
    \end{bmatrix}}}.
\end{equation*}
Similarly to~\eqref{eq:SOC_k=0_tail_psi} in Section~\ref{sec:Mode_2}, this is a collection of infinite constraints. However, for each $i$ it contains a different factor $\alpha^j_{N+i|k}$ multiplying the norm term (instead of the fixed $\sqrt{\tilde{p}^j}$), meaning that we cannot directly invoke Theorem~\ref{thm:terminal_set}. Thus, we propose to use the following terminal set for $k > 0$:
\begin{equation} \label{eq:C_Nk_def}
    \mathcal{C}_{N|k} = \myset{(z_{N|k}, \mathbf{\Phi}^\mathrm{x}_{N|k})}{z_{N|k} = \hat{z}_{N|k},\ \mathbf{\Phi}^\mathrm{x}_{N|k} = \hat{\mathbf{\Phi}}^\mathrm{x}_{N|k}},
\end{equation}
trivially satisfying~\eqref{eq:SOC_k>0_tail} $\forall j \in \mathbb{I}_{[1,c]},\ \forall i \in \mathbb{I}_{\geq 0}$ with equality.

\begin{remark} \label{rem:terminal_constraint}
    The terminal set~\eqref{eq:C_Nk_def} is more conservative than enforcing~\eqref{eq:SOC_k>0_tail} $\forall i \in \mathbb{I}_{\geq 0}$. The reason we employ this terminal set is that even if Theorem~\ref{thm:terminal_set} was extended for the case when the varying $\alpha^j_{N+i|k}$ factors take the place of $\sqrt{\tilde{p}^j}$, the obtained finite determination index $\mu$ would vary based on the $\alpha^j_{N+i|k}$ values.
    A simple approximation is to enforce~\eqref{eq:SOC_k>0_tail} $\forall i \in \mathbb{I}_{[0,\hat{\mu}]}$ instead, for a large enough number $\hat{\mu} > 0$. We compare this approximation with the terminal set~\eqref{eq:C_Nk_def} in simulations in Section~\ref{sec:numerical}.
\end{remark}
\par In summary, applying the general reconditioning SMPC framework presented in Section~\ref{sec:reconditioning} to the tractable SLP-based formulation results in the following SMPC problem (to be solved at all time steps $k>0$):
\begin{subequations} \label{eq:RHC}
    \begin{align}
        \min_{\substack{\mathbf{z}_k, \mathbf{v}_k\\ \mathbf{\Phi}^\mathrm{x}_{:|k}, \mathbf{\Phi}^\mathrm{u}_{:|k}}} \quad &J(\mathbf{z}_k, \mathbf{v}_k, \mathbf{\Phi}^\mathrm{x}_{:|k}, \mathbf{\Phi}^\mathrm{u}_{:|k})\\
        \text{s.t.} \quad\ \ 
        &z_{0|k} = x_k, \label{eq:RHC:init}\\
        &z_{i+1|k} = A z_{i|k} + B v_{i|k},\\
        &(z_{i|k}, v_{i|k}, \mathbf{\Phi}^\mathrm{x}_{i|k}, \mathbf{\Phi}^\mathrm{u}_{i|k}) \in \mathcal{C}^j_{i|k}, \label{eq:RHC:constraints_N}\\
        &(z_{N|k}, \mathbf{\Phi}^\mathrm{x}_{N|k}) \in \mathcal{C}_{N|k}, \label{eq:RHC:constraints_terminal}\\
        &\text{the SLP~\eqref{eq:SLP} is satisfied}, \\
        &\forall{j} \in \mathbb{I}_{[1,c]},\ \forall i \in \mathbb{I}_{[0,N-1]}.
    \end{align}
\end{subequations}
The cost function $J(.)$ is defined in~\eqref{eq:J}, and
\begin{equation} \label{eq:Cijk}
    \mathcal{C}^j_{i|k} := \begin{cases}
        \mathcal{C}^{j,(1)}_{i|k} \ \text{ if~\eqref{eq:pjik_degenerate} and~\eqref{eq:previous_satisfied}} \\ \mathcal{C}^{j,(2)}_{i|k} \ \text{ if~\eqref{eq:pjik_degenerate} but not~\eqref{eq:previous_satisfied}} \\
        \mathcal{C}^{j,(3)}_{i|k} \ \text{ if not~\eqref{eq:pjik_degenerate} and $\alpha^j_{i|k} \geq 0$} \\
        \mathcal{C}^{j,(4)}_{i|k} \ \text{ if not~\eqref{eq:pjik_degenerate} and $\alpha^j_{i|k} < 0$}
    \end{cases}.
\end{equation}

Similarly to existing MPC methods that optimize over the feedback matrices online~\cite{goulart2006optimization, sieber2021system, mark2022recursively}, the number of decision variables of~\eqref{eq:RHC} is $\mathcal{O}(N^2n(n+m))$, and the number of constraints is $\mathcal{O}(N(c+Nn^2))$.

\begin{remark} \label{rem:Cj1-4}
    While constraints $\mathcal{C}^{j,(1)}_{i|k}$, $\mathcal{C}^{j,(2)}_{i|k}$, and $\mathcal{C}^{j,(3)}_{i|k}$ are exact reformulations of the original chance constraint~\eqref{eq:Rec-SMPC:constraints} (meaning that they are non-conservative), $\mathcal{C}^{j,(4)}_{i|k}$ is only a sufficient condition, introducing some conservatism. However, $\mathcal{C}^{j,(4)}_{i|k}$ is only used whenever $\alpha^j_{i|k} < 0$. 
    Note that $\mathrm{Pr}[\alpha^j_{i|k} \geq 0] \geq p^j,\ \forall j \in \mathbb{I}_{[1,c]},\forall i \in \mathbb{I}_{[0,N-1]},\forall k \in \mathbb{I}_{> 0}$ and hence $\alpha^j_{i|k} < 0$ occurs with low probability.
\end{remark}

After solving optimization problem~\eqref{eq:RHC} at time step $k$, apply
\begin{equation} \label{eq:RHC:control_law}
    u_k = \pi^\star_{0|k} = v_{0|k}^\star.
\end{equation}
The proposed SMPC scheme is summarized in Algorithm~\ref{alg:SMPC}. Note that it solves~\eqref{eq:DMC} at $k=0$ and~\eqref{eq:RHC} for $k>0$. While these two optimization problems are closely related, there are differences between them. First, following the general reconditioning SMPC scheme~\eqref{eq:Rec-SMPC}, the required probability levels $p^j_{i|k}$ are defined differently for $k=0$ and $k>0$ (cf.~\eqref{eq:initial_prob} and~\eqref{eq:pjik}). Additionally, the terminal constraint of~\eqref{eq:DMC} is the maximal admissible set that exactly reformulates the original chance constraint under the terminal controller, while~\eqref{eq:RHC} uses a sufficient condition (cf. Remark~\ref{rem:terminal_constraint}). Finally, constraint~\eqref{eq:RHC:constraints_N} includes a case distinction (defined in~\eqref{eq:Cijk}, cf. Remark~\ref{rem:Cj1-4}) to formulate~\eqref{eq:Rec-SMPC:constraints} in a convex manner.

\begin{algorithm}
\caption{Proposed SMPC scheme} \label{alg:SMPC}
\begin{algorithmic}[1]
\Statex \textbf{Offline design:}
\State Choose terminal controller $K$ satisfying~\eqref{eq:b_j-norm>0}.
\State Design terminal set $\mathcal{S}_\mu$ via Algorithm~\ref{alg:terminal_set}.
\Statex \textbf{Online operation:} 
\State \textbf{Given} initial state $x_0$
\For{$k \in \mathbb{I}_{\geq 0}$}
    \State Measure $x_k$.
    \If{$k=0$}
        \State Solve~\eqref{eq:DMC}.
    \Else
        \State Update $\mathcal{C}^j_{i|k}$ $\forall i \in \mathbb{I}_{[0,N-1]},\forall j \in \mathbb{I}_{[1,c]}$ (see~\eqref{eq:Cijk}).
        \State Solve~\eqref{eq:RHC}.
    \EndIf
    \State Apply $u_k = v^\star_{0|k}$.
\EndFor
\end{algorithmic}
\end{algorithm}

\subsection{Closed-loop Analysis} \label{sec:RHC_analysis}
In this section, we show that the proposed SMPC method (Algorithm~\ref{alg:SMPC}) ensures the desired closed-loop properties.

\begin{theorem}[Closed-loop properties of Algorithm~\ref{alg:SMPC}] \label{thm:RHC}
    Suppose Assumption~\ref{ass:L} holds and~\eqref{eq:DMC} is feasible at time step $k=0$. Then, system~\eqref{eq:system} controlled by the SMPC scheme summarized in Algorithm~\ref{alg:SMPC} ensures that
    \begin{enumerate}
        \item optimization problem~\eqref{eq:RHC} remains feasible $\forall k \in \mathbb{I}_{> 0}$,
        \item chance constraints~\eqref{eq:chance_constraints} are satisfied in closed loop, and
        \item the asymptotic average cost is upper bounded by the performance of the static state-feedback controller $K$, i.e.,~\eqref{eq:lavg} is satisfied.
    \end{enumerate}
\end{theorem}
\begin{proof}
    The proof can be found in Appendix~\ref{app:RHC_proof}.
\end{proof}

The derived asymptotic average cost bound is consistent with the performance bounds of existing SMPC methods. However, compared to existing SMPC methods that use online-optimized policies~\cite{oldewurtel2008tractable, prandini2012randomized, mark2022recursively, pan2023data, li2024distributionally, knaup2024recursively}, we uniquely ensure closed-loop chance constraint satisfaction.
The improved performance of the proposed method compared to the fixed-feedback methods of~\cite{hewing2018stochastic, hewing2020recursively, kohler2022recursively, schluter2022stochastic} comes at an increased computational cost: instead of solving an SOCP, those methods require solving a QP, with significantly fewer optimization variables. 

\section{Numerical Example} \label{sec:numerical}

\par This section aims to showcase the practical applicability and the benefits of the proposed SMPC method (Algorithm~\ref{alg:SMPC}) on a building temperature control example.

The simulations are carried out using Python on a machine equipped with an Intel i7-12700H (2.30 GHz) processor with 32 GB of RAM. The SOCPs~\eqref{eq:DMC} and~\eqref{eq:RHC} and the LMIs of Algorithm~\ref{alg:terminal_set} are solved with MOSEK~\cite{mosek}, and OSQP~\cite{osqp} is used to solve QPs. In all cases, the solvers are interfaced using CVXPY~\cite{diamond2016cvxpy}.\footnote{The code is available online. \doi{10.3929/ethz-c-000802731}\\ https://gitlab.ethz.ch/ics/online-optimized-smpc}

\par \textit{Setup:} We consider the building temperature control problem adapted  from~\cite{oldewurtel2008tractable} with system matrices
\begin{equation*}
    A = \begin{bmatrix}
        0.8511 & 0.0541  & 0.0707\\ 0.1293 & 0.8635 & 0.0055 \\ 0.0989 & 0.0032 & 0.7541
    \end{bmatrix},\quad B = \begin{bmatrix}
        0.35 \\ 0.03 \\ 0.02
    \end{bmatrix},
\end{equation*}
where the first state $[x]_1$ represents the room temperature, and $[x]_2$ and $[x]_3$ are the temperatures in the walls connected with another room, and the outside, respectively. The single input variable $u$ corresponds to heating and cooling. The system is affected by disturbances $w_k \stackrel{\text{i.i.d.}}{\sim}\mathcal{N}(0, E^\top E)$, where
\begin{equation*}
    E = 10^{-3} \cdot \begin{bmatrix}
        22.2170 & 1.7912  & 42.2123 \\ 1.5376 & 0.6944 & 2.9214 \\ 103.1813 & 0.1032 & 196.0444
    \end{bmatrix},
\end{equation*}
and $[w]_1$ corresponds to the outside temperature, $[w]_2$ is the solar radiation, and $[w]_3$ represents internal heat gains. 

The goal is to minimize the heating cost, which is given by a quadratic cost on the input. Furthermore, the room temperature should be kept above 21\textdegree C with a probability of $70\%$. We shift the coordinate system such that the origin corresponds to the noise-free equilibrium at 21.5\textdegree C, and $x,u$ correspond to deviations from this equilibrium. Thus, we have cost function
\begin{equation*}
    \ell(x,u) = u^\top R  u + r^\top u,
\end{equation*}
with $R = 1$ and $r = 7$, and chance constraint
\begin{equation} \label{eq:num:constraint}
    \mathrm{Pr}[\begin{bmatrix} 1 & 0 & 0\end{bmatrix} x_k \geq -0.5] \geq 0.7,\quad \forall k \in \mathbb{I}_{\geq 0}.
\end{equation}

\par \textit{Approaches:} The following SMPC schemes relying on online-optimized feedback are compared:

\begin{itemize}
    \item Reconditioning (\textbf{RC}) SMPC : The proposed reconditioning SMPC scheme from Algorithm~\ref{alg:SMPC}.
    \item Modified Reconditioning SMPC (\textbf{RC-mod}): The same as RC, except that the terminal constraint~\eqref{eq:C_Nk_def} is replaced by directly enforcing~\eqref{eq:SOC_k>0_tail} $\forall i \in \mathbb{I}_{[0,\hat{\mu}]}$ (cf. Remark~\ref{rem:terminal_constraint}).
    \item Closed-loop Prediction (\textbf{CLP}) SMPC~\cite[Def.~6]{oldewurtel2008tractable}:
    chance constraints are enforced by computing a probability bound on the magnitude of the disturbances and then applying robust constraint satisfaction. No terminal constraint is used, soft constraints are implemented to ensure recursive feasibility, and closed-loop chance constraints are not necessarily guaranteed. 
    \item Indirect feedback (\textbf{IF}) SMPC: The combination of indirect-feedback SMPC~\cite{hewing2020recursively} with the system level parametrization~\cite{anderson2019system}, using the terminal set design proposed in Section~\ref{sec:Mode_2}.
    \item Additionally, we also carry out the simulations using the non-receding-horizon policy~\eqref{eq:dual_mode_law} for reference.
\end{itemize}

All controllers use a prediction horizon length of $N = 6$. The terminal controller is $K=0$ in all cases, yielding $ \ell_\mathrm{f} = 0$.
The terminal set is computed with Algorithm~\ref{alg:terminal_set}, resulting in $\mu = 59$.
The offline design took approximately 45 minutes, but this time can be reduced by choosing a good initial guess for $\nu$ and $\mu$. The upper bound on the indices for the terminal constraint of RC-mod is chosen as $\hat{\mu} = \mu$.
The simulation is carried out with $5\cdot 10^3$ random disturbance sequence realizations, and the simulation length is 10 time steps, starting from the initial condition $x_0 = \begin{bmatrix} 0.5 & 0 & 0 \end{bmatrix}^\top$ (i.e., from the initial room temperature of 22\textdegree C).

\begin{figure}
	\centering
	\includegraphics[scale=1]{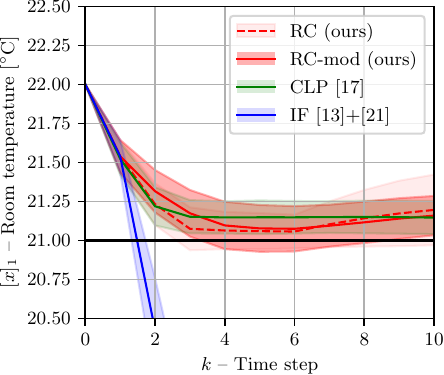}
	\caption{Closed-loop evolution of the room temperature using the two variations of the proposed method (red), the method of~\cite{oldewurtel2008tractable} (green), and the combination of indirect feedback SMPC~\cite{hewing2020recursively} with SLP~\cite{anderson2019system} (blue), for $5 \cdot 10^3$ random disturbance sequence realizations. The thick lines represent the means, and the confidence bounds correspond to $\pm$1 standard deviation.} 
 \label{fig:hvac}
\end{figure}

\par \textit{Results:} Figure~\ref{fig:hvac} presents the closed-loop evolution of the room temperature over time. As it can be clearly seen, the proposed methods RC and RC-mod are able to operate much closer to the constraints compared to CLP, by directly using the probabilistic information instead of sufficient deterministic conditions. It can also be seen that IF fails to solve the problem, as it is unable to properly address the chance constraints. 

Figure~\ref{fig:hvac} also shows a potential limitation of RC: the imposed terminal constraint~\eqref{eq:C_Nk_def} limits the SMPC from significantly deviating from the initially optimized solution for $k\geq N$, i.e., beyond the prediction horizon. 
In fact, the initially optimized policy~\eqref{eq:dual_mode_law} is very close to RC for this example. 
This is no longer the case when the terminal set is modified according to RC-mod (Remark~\ref{rem:terminal_constraint}).
Although CLP does in general not ensure closed-loop chance constraint satisfaction, they are satisfied in this example due to the absence of input constraints.

Since RC and RC-mod can approach the constraint closer than CLP, they are able to apply smaller inputs, resulting in a lower cost, as shown by Table~\ref{tab:cost}, but it can also be seen that RC can only marginally improve over the performance of policy~\eqref{eq:dual_mode_law}. The conservative nature of CLP is also highlighted by the minimal empirical constraint satisfaction level of 91.2\%, while RC and RC-mod satisfy the chance constraint tightly. Table~\ref{tab:cost} also highlights that IF greatly violates the chance constraint, and thus the low cost it achieves is meaningless. Finally, Table~\ref{tab:cost} also compares the time it took to solve the underlying optimization problems: solving the QP of CLP takes significantly less computational effort compared to the SOCPs of RC, RC-mod, and IF.

In order to investigate the effect of the required probability level $p$ on the closed-loop performance, we carry out the simulation for different values of $p$ ranging from 0.6 to 0.8, and compare RC-mod with CLP. The results are summarized in Table~\ref{tab:p_effect}. As expected, the accumulated closed-loop cost increases with the increase of $p$, and the proposed method outperforms CLP in terms of the closed-loop cost across all values of $p$.

In summary, the proposed method improves performance compared to~\cite{oldewurtel2008tractable} and also comes with stronger closed-loop chance constraint satisfaction guarantees, at the expense of a more complicated offline terminal set design and increased online computational effort.

\begin{table}
\caption{Comparison of the different methods in terms of the mean and the standard deviation of the accumulated closed-loop costs and solution times, and the minimal empirical constraint satisfaction probabilities (for $5\cdot10^3$ random disturbance sequence realizations)}
\label{tab:cost}
\begin{center}
    \begin{tabular}{@{}lccc@{}} \toprule
    {} & {Closed-loop cost} & {\shortstack{Min. satisfaction \\ level [$\%$]}} & {Solvetime [ms]} \\ \midrule
    RC (ours)  & $-18.86\pm 4.43$ & $69.4$ & $12.90\pm6.90$ \\
    RC-mod (ours)  & $-19.65\pm5.55$  & $69.7$ & $21.70\pm7.37$  \\
    CLP~\cite{oldewurtel2008tractable}  & $-18.72\pm7.54$  & $91.2$& $1.96\pm3.98$  \\ 
    IF~\cite{hewing2020recursively}+\cite{anderson2019system} & ($-115.67\pm3.49$) & $0.0$ & $27.99\pm7.27$\\
    Policy~\eqref{eq:dual_mode_law}  & $-18.85\pm4.43$  & $69.6$& $-$  \\\bottomrule
\end{tabular}
\end{center}
\end{table}

\begin{table}
\caption{The effect of $p$ on the mean and the standard deviation of the closed-loop cost accumulated by the proposed method and the method of~\cite{oldewurtel2008tractable} (for $10^4$ random disturbance sequence realizations).}
\label{tab:p_effect}
\begin{center}
    \begin{tabular}{@{}ccc@{}} \toprule
    {} & \multicolumn{2}{c}{Closed-loop cost} \\ 
    {$p$} & {RC-mod (ours)} & {CLP~\cite{oldewurtel2008tractable}} \\ \midrule
    $0.6$  & $-20.292\pm5.316$  & $-19.110\pm7.457$ \\
    $0.7$  & $-19.662\pm5.515$  & $-18.739\pm7.470$  \\
    $0.8$  & $-18.912\pm5.734$  & $-18.277\pm7.487$  \\ 
\bottomrule
\end{tabular}
\end{center}
\end{table}

\section{Conclusion} \label{sec:conclusion}

We proposed an SMPC method for linear time-invariant systems affected by additive Gaussian disturbances. The proposed method optimizes affine feedback policies online in a receding horizon fashion, by solving a finite horizon convex optimization problem at every time step. A finitely determined maximal admissible set was designed, which, combined with the idea of reconditioning on the current knowledge at every time step, guarantees closed-loop chance constraint satisfaction and recursive feasibility. The theoretical results were verified and the proposed method's applicability was demonstrated on a building temperature control example. Extension of the current results to more general problem settings including uncertain systems, nonlinear systems, output feedback, and subgaussian or bounded-variance disturbance distributions is subject of future work.

\appendix

\subsection{Proof of Lemma~\ref{lem:SOC_reform}} \label{app:SOC}
\begin{proof}
    A probabilistic half-space constraint of the form
    \begin{equation*}
        \mathrm{Pr}[h^\top y \leq g] \geq p_\mathrm{y},
    \end{equation*}
    with $y \sim \mathcal{N}(\mu_\mathrm{y}, \Sigma_\mathrm{y}),\ y \in \mathbb{R}^{n_\mathrm{y}}$, $h \in  \mathbb{R}^{n_\mathrm{y}}$, $g \in \mathbb{R}$, and $p_\mathrm{y} \in [0.5, 1)$ is equivalent to
    \begin{equation*}
        h^\top \mu_\mathrm{y} \leq g - \sqrt{\chi^2(2p_\mathrm{y}-1)}\sqrt{h^\top \Sigma_\mathrm{y} h},
    \end{equation*}
    see, e.g.,~\cite{hewing2020recursively}. Applying this result to~\eqref{eq:chance_constraint_j} combined with~\eqref{eq:SLP:ux} results in~\eqref{eq:SOC_k=0_Psi}. Furthermore, using
    \begin{equation*}
        \mathrm{vec}(ABC) = (C^\top \otimes A) \mathrm{vec}(B),
    \end{equation*}
    we obtain
    \begin{equation*}
    \begin{split}
       \mathbf{\Sigma}^{\mathrm{w},1/2}_i& \begin{bmatrix} (\mathbf{\Phi}^{\mathrm{x}}_{i|0})^\top & (\mathbf{\Phi}^{\mathrm{u}}_{i|0})^\top\end{bmatrix} \begin{bmatrix} G_j^\top \\ H_j^\top \end{bmatrix} \\ &= \mathrm{vec}\left((\mathbf{\Sigma}^{\mathrm{w}}_i)^{1/2} \begin{bmatrix} (\mathbf{\Phi}^{\mathrm{x}}_{i|0})^\top & (\mathbf{\Phi}^{\mathrm{u}}_{i|0})^\top\end{bmatrix} \begin{bmatrix} G_j^\top \\ H_j^\top \end{bmatrix}\right) \\
        &= \left(\begin{bmatrix} G_j & H_j \end{bmatrix} \otimes (\mathbf{\Sigma}^{\mathrm{w}}_i)^{1/2}\right) \mathrm{vec}\left(\begin{bmatrix} (\mathbf{\Phi}^{\mathrm{x}}_{i|0})^\top & (\mathbf{\Phi}^{\mathrm{u}}_{i|0})^\top\end{bmatrix}\right) \\
        & = \left(\begin{bmatrix} G_j & H_j \end{bmatrix} \otimes (\mathbf{\Sigma}^{\mathrm{w}}_i)^{1/2}\right) \psi_{i|0},
    \end{split}
    \end{equation*}
    and thus the equivalence of~\eqref{eq:SOC_k=0} and~\eqref{eq:SOC_k=0_Psi} is proven.
\end{proof}

\subsection{Terminal controller synthesis satisfying Assumption~\ref{ass:L}} \label{app:K_design}
In order to formulate a convex optimization problem that can be solved in practice, the strict inequality~\eqref{eq:b_j-norm>0} in Assumption~\ref{ass:L} is approximated by $b_j -  \sqrt{\tilde{p}^j}\norm{(\Sigma_\infty^{\mathrm{x}})^{1/2} C_K^\top L_j^\top} \geq \varepsilon_j$, where $\varepsilon_j \in \mathbb{R}_{>0}$ is sufficiently small. Then, using the Schur complement, we can formulate the following convex semidefinite program (SDP) to design $K$:
\begin{subequations} \label{eq:K_design}
    \begin{align}
        \min_{\Sigma^{\mathrm{x}}_\infty,Y,t} \quad &t \\
        \mathrm{s.t.} \quad\ &\Sigma^{\mathrm{x}}_\infty = (\Sigma^{\mathrm{x}}_\infty)^\top \preceq tI_n, \\
        &\begin{bmatrix}
            \Sigma^{\mathrm{x}}_\infty - \Sigma^\mathrm{w} & A\Sigma^{\mathrm{x}}_\infty + BY \\ (A\Sigma^{\mathrm{x}}_\infty + BY)^\top & \Sigma^{\mathrm{x}}_\infty
        \end{bmatrix} \succeq 0, \\
         &\begin{bmatrix}
            (b_j - \varepsilon_j)^2 & L_j(C\Sigma^{\mathrm{x}}_\infty + DY) \\ (C\Sigma^{\mathrm{x}}_\infty + DY)^\top L_j^\top & (1/\tilde{p}_j)\Sigma^{\mathrm{x}}_\infty
        \end{bmatrix} \succeq 0, \\
        &\forall j \in \mathbb{I}_{[1,c]},
    \end{align}
\end{subequations}
where $t \in \mathbb{R}_{\geq 0}$, $\Sigma^{\mathrm{x}}_\infty \in \mathbb{R}^{n\times n}$, $Y \in \mathbb{R}^{m\times n}$. Infeasibility of~\eqref{eq:K_design} means that the desired probability level $p^j$ is too high for the given constraint and disturbance covariance. If~\eqref{eq:K_design} is feasible, a suitable choice for $K$ satisfying Assumption~\ref{ass:L} is
\begin{equation*}
    K = Y^\star (\Sigma^{\mathrm{x}\star}_\infty)^{-1}.
\end{equation*}

\subsection{Proof of Theorem~\ref{thm:terminal_set}} \label{app:terminal_set_proof}
    Throughout the proof, we will make use of the relaxed version of inequality~\eqref{eq:SOC_k=0_tail_psi} and set $\mathcal{D}_i$:
    \begin{equation}
         G_{K,j} A_K^i z \leq b_j - \sqrt{\tilde{p}^j} \norm{
            \left( G_{K,j} A_K^i \otimes (\mathbf{\Sigma}^{\mathrm{w}}_N)^{1/2}\right) \psi}, \label{eq:SOC_k=0_tail_psi_relaxed}
    \end{equation}
    which corresponds to~\eqref{eq:SOC_k=0_tail_psi} with $\Sigma^\mathrm{x}_i = 0$, and
    \begin{equation*}
         \hat{\mathcal{D}}_i = \myset{(z, \psi)}{\eqref{eq:SOC_k=0_tail_psi_relaxed} \text{ holds for }i,\ \forall j \in \mathbb{I}_{[1, c]}};
    \end{equation*}
    furthermore, $\hat{\mathcal{S}}_i$ and $\hat{\mathcal{S}}_\infty$ can be defined analogously to $\mathcal{S}_i$ and $\mathcal{S}_\infty$. Then, it is easy to see that
    \begin{equation} \label{eq:hat_tilde_order}
        \bar{\mathcal{D}}_i \subseteq \mathcal{D}_i \subseteq \hat{\mathcal{D}}_i, \quad \bar{\mathcal{S}}_i \subseteq \mathcal{S}_i \subseteq \hat{\mathcal{S}}_i, \quad \bar{\mathcal{S}}_\infty \subseteq \mathcal{S}_\infty \subseteq \hat{\mathcal{S}}_\infty.
    \end{equation}
    Finally, consider inequality
    \begin{equation} \label{eq:proofs:L_hat_ineq}
        L_jy + \sqrt{\tilde{p}^j}\norm{(L_j \otimes (\mathbf{\Sigma}^{\mathrm{w}}_N)^{1/2})\tilde{y}}\leq b_j,
    \end{equation}
    where $y \in \mathbb{R}^d$, $\tilde{y} \in \mathbb{R}^{Nnd}$, and define the set
    \begin{equation} \label{eq:proofs:L_hat}
        \hat{\mathcal{L}} = \myset{\begin{bmatrix} y \\\tilde{y}\end{bmatrix} \in \mathbb{R}^{d+Nnd}}{\eqref{eq:proofs:L_hat_ineq}\text{ holds }\forall j \in \mathbb{I}_{[1,c]}}.
    \end{equation}
\begin{proof}    
    The proof extends arguments from the maximal admissible set characterizations presented in~\cite{gilbert1991linear, li2021chance}. In Part I, we establish the boundedness of $\hat{\mathcal{L}}$, which is then combined with observability and Schur stability to prove finite determination and boundedness of $\hat{\mathcal{S}}_\infty$ in Part II. Finally, Part III proves the finite determination of $\mathcal{S}_\infty$ based on $\hat{\mathcal{S}}_\infty$.
    \par \textit{Part I: Boundedness of $\hat{\mathcal{L}}$.} It can be seen that~\eqref{eq:proofs:L_hat_ineq} is equivalent to the following pair of inequalities:
    \begin{equation} \label{eq:proofs:L_hat_ineq_pair_1}
        \tilde{p}^j\tilde{y}^\top(L_j^\top L_j \otimes \mathbf{\Sigma}^\mathrm{w}_N)\tilde{y} \leq (b_j-L_jy)^2,\quad L_jy \leq b_j.
    \end{equation}
    Using the fact that 
    \begin{equation*}
        L_j^\top L_j \otimes \mathbf{\Sigma}^\mathrm{w}_N = \mathcal{P}^\top(\mathbf{\Sigma}^\mathrm{w}_N \otimes L_j^\top L_j)\mathcal{P},
    \end{equation*}
    where $\mathcal{P}$ is a permutation matrix (see, e.g.,~\cite{henderson1981vec}), and introducing
    \begin{equation} \label{eq:eta_def}
        \eta := ((\mathbf{\Sigma}^{\mathrm{w}}_N)^{1/2} \otimes I_d)\mathcal{P}\tilde{y},
    \end{equation}
    inequality~\eqref{eq:proofs:L_hat_ineq_pair_1} can be further written as
    \begin{equation} \label{eq:proofs:L_hat_ineq_pair_2}
        \eta^\top(\tilde{p}^jI_{Nn} \otimes L_j^\top L_j)\eta \leq (b_j-L_jy)^2,\quad L_jy \leq b_j.
    \end{equation}
    Note that $((\mathbf{\Sigma}^{\mathrm{w}}_N)^{1/2} \otimes I_d)\mathcal{P}$ is invertible, since $\mathcal{P}$ is orthogonal and $\mathbf{\Sigma}^{\mathrm{w}}_N \succ 0$. Taking the partition
    \begin{equation*}
        \eta^\top = \begin{bmatrix} \eta_1^\top & ... & \eta_{Nn}^\top\end{bmatrix},\quad \eta_i \in \mathbb{R}^d,\ \forall i \in \mathbb{I}_{[1,Nn]},
    \end{equation*}
    conditions~\eqref{eq:proofs:L_hat_ineq_pair_2} can be written as
    \begin{equation*} \label{eq:proofs:L_hat_ineq_pair_3}
        \tilde{p}^j\sum_{i=1}^{Nn}\eta_i^\top L_j^\top L_j \eta_i \leq (b_j-L_jy)^2,\quad L_jy \leq b_j,
    \end{equation*}
    which implies
    \begin{equation} \label{eq:proofs:L_hat_ineq_pair_4}
        \tilde{p}^j(L_j \eta_i)^2 \leq (b_j-L_jy)^2,\ \forall i \in \mathbb{I}_{[1,Nn]},\quad L_jy \leq b_j,
    \end{equation}
    due to $L_j^\top L_j \succeq 0$. Finally,~\eqref{eq:proofs:L_hat_ineq_pair_4} further implies
    \begin{equation*}
        \sqrt{\tilde{p}^j}L_j\eta_i \leq b_j - L_j y,\ \forall i \in \mathbb{I}_{[1,Nn]},\quad L_jy \leq b_j.
    \end{equation*}
    We know that if $\begin{bmatrix}y^\top & \tilde{y}^\top\end{bmatrix}^\top \in \hat{\mathcal{L}}$, then $Ly\leq b$, i.e., $y \in \mathcal{L}$, which is bounded, implying that $\underline{b}_j = \min_{y\in \mathcal{L}} L_jy$ is a well-defined finite number; furthermore, due to $b > 0$, $\underline{b}_j \leq 0$. Therefore, if $\begin{bmatrix}y^\top & \tilde{y}^\top \end{bmatrix}^\top \in \hat{\mathcal{L}}$, then 
    \begin{equation*}
        \sqrt{\tilde{p}^j}L_j\eta_i \leq \max_{y\in \mathcal{L}}(b_j - L_j y),\ \forall i \in \mathbb{I}_{[1,Nn]},\ \forall j \in \mathbb{I}_{[1,c]},
    \end{equation*}
    i.e.,
    \begin{equation*}
        \sqrt{\tilde{p}^j}L_j\eta_i \leq b_j - \underline{b}_j,\ \forall i \in \mathbb{I}_{[1,Nn]},\ \forall j \in \mathbb{I}_{[1,c]}.
    \end{equation*}
    Define $\bar{b} \in \mathbb{R}^c_{> 0}$, the $j$-th row of which is $\bar{b}_j = 1/\sqrt{\tilde{p}^j} (b_j - \underline{b}_j) > 0$. Since $\mathcal{L}$ is bounded, the set $\myset{\eta_i}{L\eta_i \leq \bar{b}}$ is also bounded $\forall i \in \mathbb{I}_{[1, Nn]}$. Thus, $\exists M \in \mathbb{R}_{>0}$ such that $\norm{\eta} \leq M,\ \forall \begin{bmatrix}y^\top & \tilde{y}^\top \end{bmatrix}^\top \in \hat{\mathcal{L}} $, i.e., according to~\eqref{eq:eta_def},
    \begin{equation*}
        \norm{((\mathbf{\Sigma}^{\mathrm{w}}_N)^{1/2} \otimes I_d)\mathcal{P}\tilde{y}} \leq M,
    \end{equation*}
    and thus
     \begin{equation*}
        \begin{split}
        \norm{\tilde{y}} &= \norm{\left(((\mathbf{\Sigma}^{\mathrm{w}}_N)^{1/2} \otimes I_d)\mathcal{P}\right)^{-1}((\mathbf{\Sigma}^{\mathrm{w}}_N)^{1/2} \otimes I_d)\mathcal{P}\tilde{y}} \\
        &\leq \norm{\left(((\mathbf{\Sigma}^{\mathrm{w}}_N)^{1/2} \otimes I_d)\mathcal{P}\right)^{-1}} M.
        \end{split}
    \end{equation*}
    That is, both $y$ and $\tilde{y}$ are bounded, proving that the set $\hat{\mathcal{L}}$ is bounded.
    \par \textit{Part II: Finite determination of $\hat{\mathcal{S}}_\infty$.} Note that $L_jC_K = G_{K,j}$ due to $G = LC$ and $H = LD$, and thus~\eqref{eq:SOC_k=0_tail_psi_relaxed} can be equivalently written as
    \begin{equation*} \label{eq:proofs:L_hat_ineq_CA}
        L_jC_KA_K^iz + \sqrt{\tilde{p}^j}\norm{(L_jC_KA_K^i \otimes(\mathbf{\Sigma}^{\mathrm{w}}_N)^{1/2})\psi}\leq b_j.
    \end{equation*}
    Introducing the notation
    \begin{align*}
        \zeta^\top &= \begin{bmatrix}z^\top & \psi^\top \end{bmatrix} \in \mathbb{R}^{Nn^2+n}, \\
        \mathbf{A} &= \begin{bmatrix}
            A_K & \\ & A_K \otimes I_{Nn}
        \end{bmatrix} \in \mathbb{R}^{(Nn+1)n \times (Nn+1)n}, \\
        \mathbf{C} &= \begin{bmatrix}
            C_K & \\ & C_K \otimes I_{Nn}
        \end{bmatrix} \in \mathbb{R}^{(Nn+1)d \times (Nn+1)n},
    \end{align*}
    it can be seen that $(z, \psi)\in \hat{\mathcal{D}}_i \iff \mathbf{C} \mathbf{A}^i \zeta \in \hat{\mathcal{L}}$ and thus $(z, \psi)\in \hat{\mathcal{S}}_k \iff \mathbf{C} \mathbf{A}^i \zeta \in \hat{\mathcal{L}},\ \forall i \in \mathbb{I}_{[1,k]}$. Furthermore,
    \begin{align*}
        \mathbf{C} \mathbf{A}^i &= \begin{bmatrix}
            C_KA_K^i & \\ & C_KA_K^i \otimes I_{Nn}
        \end{bmatrix} \nonumber \\
        &= \begin{bmatrix}
            C_KA_K^i & \\ & \mathcal{P}_1(I_{Nn}  \otimes C_KA_K^i)\mathcal{P}_2
        \end{bmatrix} \nonumber \\
        &= \begin{bmatrix}
            I_d & \\ & \mathcal{P}_1
        \end{bmatrix}\begin{bmatrix}
            C_KA_K^i & \\ & I_{Nn}  \otimes C_KA_K^i
        \end{bmatrix} \begin{bmatrix}
            I_n & \\ & \mathcal{P}_2
        \end{bmatrix} \nonumber \\
        &= \mathcal{P}_3(I_{Nn+1}  \otimes C_KA_K^i)\mathcal{P}_4 \nonumber \\
        &= \mathcal{P}_3\mathcal{P}_5(C_KA_K^i \otimes I_{Nn+1})\mathcal{P}_6\mathcal{P}_4 \nonumber \\
        &= \mathcal{P}_7(C_KA_K^i \otimes I_{Nn+1})\mathcal{P}_8, 
    \end{align*}
    where $\mathcal{P}_j, j \in \mathbb{I}_{[1,8]}$ are permutation matrices that depend only on the dimensions of $C_KA_K^i$ and $I_{Nn}$, i.e., they are the same for every $i$. Define observability matrices $\mathcal{O}^\top = \begin{bmatrix} C_K^\top & A_K^\top C_K^\top & ... & A_K^{n-1\top} C_K^\top\end{bmatrix}$ and
    \begin{equation*}
        \begin{split}
            \hat{\mathcal{O}} &= \begin{bmatrix}
                \mathbf{C} \\ \mathbf{C}\mathbf{A} \\ \vdots \\ \mathbf{C}\mathbf{A}^{n-1}
            \end{bmatrix} =
            \begin{bmatrix}
                \mathcal{P}_7(C_K \otimes I_{Nn+1})\mathcal{P}_8 \\ \mathcal{P}_7(C_KA_K \otimes I_{Nn+1})\mathcal{P}_8 \\ \vdots \\ \mathcal{P}_7(C_KA_K^{n-1} \otimes I_{Nn+1})\mathcal{P}_8
            \end{bmatrix}\\
            &= (I_n \otimes \mathcal{P}_7) \begin{bmatrix}
                C_K \otimes I_{Nn+1} \\ C_KA_K \otimes I_{Nn+1} \\ \vdots \\ C_KA_K^{n-1} \otimes I_{Nn+1}
            \end{bmatrix} \mathcal{P}_8 \\
            &= (I_n \otimes \mathcal{P}_7) (\mathcal{O} \otimes I_{Nn+1}) \mathcal{P}_8.
        \end{split}
    \end{equation*}
    As $(C_K, A_K)$ is observable, $\mathrm{rank}(\mathcal{O}) = n$. Since row and column permutations do not change the rank,
    \begin{equation*}
        \begin{split}
            \mathrm{rank}(\hat{\mathcal{O}}) &= \mathrm{rank}(\mathcal{O} \otimes I_{Nn+1}) \\ & =\mathrm{rank}(\mathcal{O})\mathrm{rank}(I_{Nn+1}) = n(Nn+1),
        \end{split}
    \end{equation*}
    proving that $(\mathbf{C},\mathbf{A})$ is observable. Furthermore, $\mathbf{A}$ is Schur due to $A_K$ being Schur. This, combined with the boundedness of $\hat{\mathcal{L}}$ and the fact that $0 \in \mathrm{int}(\hat{\mathcal{L}})$ (owing to $b > 0$), implies that $\hat{\mathcal{S}}_\infty$ is bounded and it is finitely determined, according to~\cite[Theorems~2.1 and 4.1]{gilbert1991linear}. That is, $\exists \mu_0 \in \mathbb{I}_{\geq 0}$ such that $\hat{\mathcal{S}}_{\mu_0} = \hat{\mathcal{S}}_\infty$ and $\hat{\mathcal{S}}_{\mu_0}$ is bounded.
    \par \textit{Part III: Finite determination of $\mathcal{S}_\infty$.} From~\eqref{eq:hat_tilde_order}, the boundedness of $\hat{\mathcal{S}}_{\mu_0}$ implies that $\mathcal{S}_\mu \subseteq \mathcal{S}_{\mu_0} \subseteq \hat{\mathcal{S}}_{\mu_0}$ is also bounded $\forall \mu \in \mathbb{I}_{\geq \mu_0}$. Combining this result with $\lim_{i \rightarrow \infty} A_K^i = 0$, it is true that $\forall \varepsilon_j \in \mathbb{R}_{>0}\ \exists \mu \in \mathbb{I}_{\geq\mu_0}$ such that
    \begin{align*}
        G&_{K,j}A_K^i z + \sqrt{\tilde{p}^j} \norm{\begin{bmatrix}
            \left( G_{K,j} A_K^i \otimes (\mathbf{\Sigma}^{\mathrm{w}}_N)^{1/2}\right) \psi\\ (\Sigma_\infty^{\mathrm{x}})^{1/2} G_{K,j}^\top
        \end{bmatrix}} \\
        &- \sqrt{\tilde{p}^j}\norm{(\Sigma_\infty^{\mathrm{x}})^{1/2} G_{K,j}^\top} \leq \varepsilon_j,\quad \forall (z,\psi) \in \mathcal{S}_\mu,\ \forall i \in \mathbb{I}_{> \mu}. \nonumber
    \end{align*}
    Thus, for $\varepsilon_j = b_j - \sqrt{\tilde{p}^j}\norm{(\Sigma_\infty^{\mathrm{x}})^{1/2} G_{K,j}^\top} > 0: \exists \mu \in \mathbb{I}_{\geq \mu_0}$ such that
    \begin{equation*}
        G_{K,j} A_K^i z + \sqrt{\tilde{p}^j} \norm{\begin{bmatrix}
            \left( G_{K,j} A_K^i \otimes (\mathbf{\Sigma}^{\mathrm{w}}_N)^{1/2}\right) \psi\\ (\Sigma_\infty^{\mathrm{x}})^{1/2} G_{K,j}^\top
        \end{bmatrix}} \leq b_j,
    \end{equation*}
    $\forall i \in \mathbb{I}_{> \mu},\ \forall j \in \mathbb{I}_{[1,c]},\ \forall (z,\psi) \in \mathcal{S}_\mu$, in other words,
    \begin{equation*}
        \mathcal{S}_\mu \subseteq \bigcap_{i=\mu+1}^\infty \bar{\mathcal{D}}_i \stackrel{\eqref{eq:hat_tilde_order}}{\subseteq} \bigcap_{i=\mu+1}^\infty \mathcal{D}_i.
    \end{equation*}
    Consequently,
    \begin{equation}
        \mathcal{S}_\infty = \mathcal{S}_\mu \cap  \left(\bigcap_{i=\mu+1}^\infty \mathcal{D}_i\right) \supseteq \mathcal{S}_\mu \cap \mathcal{S}_\mu = \mathcal{S}_\mu,
    \end{equation}
    i.e., $\mathcal{S}_\infty \supseteq \mathcal{S}_\mu$. On the other hand, $\mathcal{S}_\infty \subseteq\mathcal{S}_\mu$ by construction, resulting in $\mathcal{S}_\infty = \mathcal{S}_\mu$, concluding the proof.
\end{proof}

\subsection{Proof of Proposition~\ref{prop:terminal_set_alg}} \label{app:algorithm_proof}

\begin{proof}
The idea behind Algorithm~\ref{alg:terminal_set} and the proof is the following. Step~\ref{algstep:nu} finds the smallest index $\nu$ that defines the maximal admissible set with respect to the tightened constraint~\eqref{eq:SOC_k=0_tail_psi_tightened}. Consequently, Step~\ref{algstep:mu} then gives the smallest index $\mu$ such that imposing the original time-varying constraint~\eqref{eq:SOC_k=0_tail_psi} $\forall i \in \mathbb{I}_{[0,\mu]}$ implies that~\eqref{eq:SOC_k=0_tail_psi_tightened} is satisfied $\forall i \in \mathbb{I}_{[\mu+1,\mu+\nu+1]}$, which in turn also implies satisfaction $\forall i \in \mathbb{I}_{\geq \mu+1}$ (due to Step~\ref{algstep:nu}), meaning that $\mu$ satisfies~\eqref{eq:terminal_set_Smu_inclusion}.

Let
    \begin{equation} \label{eq:proofs:L_tilde_ineq}
        L_jy + \sqrt{\tilde{p}^j}\norm{\begin{bmatrix}
            (L_j \otimes (\mathbf{\Sigma}^{\mathrm{w}}_N)^{1/2})\tilde{y} \\ (\Sigma^{\mathrm{x}}_\infty)^{1/2} C_K^\top L_j^\top
        \end{bmatrix}} \leq b_j,
    \end{equation}
    \begin{equation} \label{eq:proofs:L_tilde}
        \bar{\mathcal{L}} = \myset{\begin{bmatrix} y \\\tilde{y}\end{bmatrix} \in \mathbb{R}^{d+Nnd}}{\eqref{eq:proofs:L_tilde_ineq}\text{ holds }\forall j \in \mathbb{I}_{[1,c]}},
    \end{equation}
    be defined analogously to~\eqref{eq:proofs:L_hat_ineq}-\eqref{eq:proofs:L_hat}, satisfying
    \begin{equation*}
        \mathbf{C}\mathbf{A}^i\zeta \in \bar{\mathcal{L}} \iff (z, \psi) \in \bar{\mathcal{D}}_i,
    \end{equation*}
    that is,
    \begin{equation*}
        \mathbf{C}\mathbf{A}^i\zeta \in \bar{\mathcal{L}},\ \forall i \in \mathbb{I}_{[0,\nu]} \iff (z, \psi) \in \bar{\mathcal{S}}_\nu.
    \end{equation*}
    Note that $\bar{\mathcal{L}} \subseteq \hat{\mathcal{L}}$ is bounded due to $\hat{\mathcal{L}}$ being bounded. Moreover, $0 \in \mathrm{int}(\bar{\mathcal{L}})$ due to Assumption~\ref{ass:L}. Finally, $(\mathbf{C}, \mathbf{A})$ is observable and $\mathbf{A}$ is Schur, and thus~\cite[Theorem~4.1]{gilbert1991linear} guarantees that $\exists \nu \in \mathbb{I}_{\geq 0}$ such that $\bar{\mathcal{S}}_\infty = \bar{\mathcal{S}}_\nu$.
    In other words,
    \begin{equation} \label{eq:nu_equivalence}
        \mathbf{C}\mathbf{A}^i\zeta \in \bar{\mathcal{L}},\ \forall i \in \mathbb{I}_{[0,\nu]} \iff \mathbf{C}\mathbf{A}^i\zeta \in \bar{\mathcal{L}},\ \forall i \in \mathbb{I}_{\geq 0}.
    \end{equation}
    According to~\cite[Algorithm~3.1]{gilbert1991linear}, $\nu$ is the smallest non-negative integer for which $\bar{\mathcal{S}}_\nu \subseteq \bar{\mathcal{S}}_{\nu+1}$, i.e., we can find $\nu$ by finding the smallest non-negative integer satisfying
    \begin{equation} \label{eq:nu_containment}
         \bar{\mathcal{S}}_\nu \subseteq \bar{\mathcal{D}}_{\nu+1},
    \end{equation}
    which is exactly Step~\ref{algstep:nu} of Algorithm~\ref{alg:terminal_set}.
    \par After obtaining $\nu$, we know that~\eqref{eq:nu_equivalence} holds for any $\zeta \in \mathbb{R}^{Nn^2+n}$, meaning that it also holds for $\mathbf{A}^{\mu+1}\xi,\ \xi \in \mathbb{R}^{Nn^2+n}$ (where $\mu \in \mathbb{I}_{\geq 0}$):
    \begin{equation*} \label{eq:mu_nu_equivalence}
        \mathbf{C}\mathbf{A}^i\xi \in \bar{\mathcal{L}},\forall i \in \mathbb{I}_{[\mu+1,\mu+\nu+1]} \iff \mathbf{C}\mathbf{A}^i\xi \in \bar{\mathcal{L}},\forall i \in \mathbb{I}_{\geq \mu+1},
    \end{equation*}
    or equivalently,
    \begin{equation*}
        \bigcap_{i=\mu+1}^{\nu+\mu+1} \bar{\mathcal{D}}_i = \bigcap_{i=\mu+1}^\infty \bar{\mathcal{D}}_i.
    \end{equation*}
    Consequently, the set inclusion
    \begin{equation} \label{eq:mu_containment}
        \mathcal{S}_\mu \subseteq \bigcap_{i=\mu+1}^{\nu+\mu+1} \bar{\mathcal{D}}_i
    \end{equation}
    in Step~\ref{algstep:mu} of Algorithm~\ref{alg:terminal_set} is equivalent to checking the original set inclusion~\eqref{eq:terminal_set_Smu_inclusion} in Theorem~\ref{thm:terminal_set}, concluding the proof.
\end{proof}

\subsection{S-procedure to find \texorpdfstring{$\nu$}{v} in Step~\ref{algstep:nu} of Algorithm~\ref{alg:terminal_set}} \label{app:finding_nu}
In the following, it is shown how to use the S-procedure to formulate a sufficient condition for the set containment~\eqref{eq:nu_containment}. For any given $i$, $\mathbf{C}\mathbf{A}^i\zeta \in \bar{\mathcal{L}}$ is 
equivalent to the following pair of inequalities:
\begin{equation} \label{eq:L_quadratic}
    \begin{split} 
        \zeta^\top (\mathbf{A}^{i})^\top \mathbf{C}^\top F_j \mathbf{C} \mathbf{A}^i \zeta + 2f_j^\top\mathbf{C}\mathbf{A}^i \zeta + \varphi_j &\geq 0, \\
        2g_j^\top \mathbf{C} \mathbf{A}^i \zeta + b_j &\geq 0,
    \end{split}
\end{equation}
$\forall j \in \mathbb{I}_{[1,c]}$, where
\begin{equation} \label{eq:Ffg}
    \begin{split}
        F_j &= \begin{bmatrix}
        L_j^\top L_j & \\ & -\tilde{p}^j L_j^\top L_j \otimes \mathbf{\Sigma}^\mathrm{w}_N
        \end{bmatrix},\\
        f_j^\top &= \begin{bmatrix}
            -b_jL_j & 0
        \end{bmatrix}, \\
        \varphi_j &= b_j^2 - \tilde{p}^jL_jC_K\Sigma^\mathrm{x}_\infty C_K^\top L_j^\top, \\
        g_j^\top &= \frac{1}{2}\begin{bmatrix} -L_j & 0
        \end{bmatrix}.
    \end{split}
\end{equation}
Thus, the goal is to find $\nu$ such that~\eqref{eq:L_quadratic} $ \forall i \in \mathbb{I}_{[0, \nu]},\, \forall j \in \mathbb{I}_{[1,c]}$ implies~\eqref{eq:L_quadratic} for $i=\nu+1,\, \forall j \in \mathbb{I}_{[1,c]}$.
According to the S-procedure for quadratic functions and nonstrict inequalities (see~\cite[Section~2.6.3]{boyd1994linear}), a sufficient condition for this implication is the existence of non-negative scalars $\alpha_{k,l}^{(j)}, \beta_{k,l}^{(j)}, \gamma_{k,l}^{(j)}, \delta_{k,l}^{(j)}$ ($k \in \mathbb{I}_{[0,\nu]}$, $l \in \mathbb{I}_{[1,c]}$) such that the following LMIs hold for all $j \in \mathbb{I}_{[1,c]}$:
\begin{equation*} \label{eq:nu_LMI_f}
    \begin{split}
        &\begin{bmatrix}
            (\mathbf{A}^{\nu+1})^\top \mathbf{C}^\top F_j \mathbf{C} \mathbf{A}^{\nu+1} & \mathbf{A}^{\nu+1\top} \mathbf{C}^\top f_j\\ f_j^\top\mathbf{C}\mathbf{A}^{\nu+1} & \varphi_j
        \end{bmatrix} \\
        &\ -\sum_{k=0}^\nu \sum_{l=1}^c \alpha_{k,l}^{(j)} \begin{bmatrix}
            (\mathbf{A}^{k})^\top \mathbf{C}^\top F_l \mathbf{C} \mathbf{A}^k & (\mathbf{A}^{k})^\top \mathbf{C}^\top f_l\\ f_l^\top\mathbf{C}\mathbf{A}^k & \varphi_l
        \end{bmatrix} \\
        &\ -\sum_{k=0}^\nu \sum_{l=1}^c \beta_{k,l}^{(j)} \begin{bmatrix}
            0 & (\mathbf{A}^{k})^\top \mathbf{C}^\top g_l\\ g_l^\top\mathbf{C}\mathbf{A}^k & b_l
        \end{bmatrix}\succeq 0,
        \end{split}
\end{equation*}
and
\begin{equation*} \label{eq:nu_LMI_g}
    \begin{split}
        &\begin{bmatrix}
            0 & (\mathbf{A}^{\nu+1})^\top \mathbf{C}^\top g_j\\ g_j^\top\mathbf{C}\mathbf{A}^{\nu+1} & b_j
        \end{bmatrix} \\
        &\ -\sum_{k=0}^\nu \sum_{l=1}^c \gamma_{k,l}^{(j)} \begin{bmatrix}
            (\mathbf{A}^{k})^\top \mathbf{C}^\top F_l \mathbf{C} \mathbf{A}^k & (\mathbf{A}^{k})^\top \mathbf{C}^\top f_l\\ f_l^\top\mathbf{C}\mathbf{A}^k & \varphi_l
        \end{bmatrix} \\
        &\ -\sum_{k=0}^\nu \sum_{l=1}^c \delta_{k,l}^{(j)} \begin{bmatrix}
            0 & (\mathbf{A}^{k})^\top \mathbf{C}^\top g_l\\ g_l^\top\mathbf{C}\mathbf{A}^k & b_l
        \end{bmatrix}\succeq 0.
        \end{split}
\end{equation*}

\subsection{S-procedure to find \texorpdfstring{$\mu$}{u} in Step~\ref{algstep:mu} of Algorithm~\ref{alg:terminal_set}} \label{app:finding_mu}
Similarly to~\eqref{eq:L_quadratic}, $(z, \psi) \in \mathcal{S}_\mu$ is equivalent to the following pair of inequalities:
\begin{equation} \label{eq:mu_L_quadratic}
    \begin{split} 
        \zeta^\top (\mathbf{A}^{i})^\top \mathbf{C}^\top F_j \mathbf{C} \mathbf{A}^i \zeta + 2f_j^\top\mathbf{C}\mathbf{A}^i \zeta + \phi_{i,j} &\geq 0, \\
        2g_j^\top \mathbf{C} \mathbf{A}^i \zeta + b_j &\geq 0,
    \end{split}
\end{equation}
$\forall j \in \mathbb{I}_{[1,c]}$ and $\forall i \in \mathbb{I}_{[0,\mu]}$, where $F_j$, $f_j$, and $g_j$ are defined in~\eqref{eq:Ffg}, and
\begin{equation*}
    \phi_{i,j} = b_j^2 - \tilde{p}^j L_jC_K\Sigma^\mathrm{x}_iC_K^\top L_j^\top.
\end{equation*}
Once again, the S-procedure provides a sufficient condition for~\eqref{eq:mu_L_quadratic}. If there exist non-negative scalars $\alpha_{k,l}^{(i,j)}, \beta_{k,l}^{(i,j)}, \gamma_{k,l}^{(i,j)}, \delta_{k,l}^{(i,j)}$ ($k \in \mathbb{I}_{[0,\mu]}$, $l \in \mathbb{I}_{[1,c]}$) such that the following LMIs hold $\forall j \in \mathbb{I}_{[1,c]}, \forall i \in \mathbb{I}_{[\mu+1,\mu+\nu+1]}$
\begin{equation*} \label{eq:mu_LMI_f}
    \begin{split}
        &\begin{bmatrix}
            (\mathbf{A}^{i})^\top \mathbf{C}^\top F_j \mathbf{C} \mathbf{A}^i & (\mathbf{A}^{i})^\top \mathbf{C}^\top f_j\\ f_j^\top\mathbf{C}\mathbf{A}^i & \varphi_j
        \end{bmatrix} \\
        &\ -\sum_{k=0}^\mu \sum_{l=1}^c \alpha_{k,l}^{(i,j)} \begin{bmatrix}
            (\mathbf{A}^{k})^\top \mathbf{C}^\top F_l \mathbf{C} \mathbf{A}^k & (\mathbf{A}^{k})^\top \mathbf{C}^\top f_l\\ f_l^\top\mathbf{C}\mathbf{A}^k & \phi_{l,k}
        \end{bmatrix} \\
        &\ -\sum_{k=0}^\mu \sum_{l=1}^c \beta_{k,l}^{(i,j)} \begin{bmatrix}
            0 & (\mathbf{A}^{k})^\top \mathbf{C}^\top g_l\\ g_l^\top\mathbf{C}\mathbf{A}^k & b_l
        \end{bmatrix}\succeq 0,
        \end{split}
\end{equation*}
and
\begin{equation*} \label{eq:mu_LMI_g}
    \begin{split}
        &\begin{bmatrix}
            0 & (\mathbf{A}^{i})^\top \mathbf{C}^\top g_j\\ g_j^\top\mathbf{C}\mathbf{A}^i & b_j
        \end{bmatrix} \\
        &\ -\sum_{k=0}^\mu \sum_{l=1}^c \gamma_{k,l}^{(i,j)} \begin{bmatrix}
            (\mathbf{A}^{k})^\top \mathbf{C}^\top F_l \mathbf{C} \mathbf{A}^k & (\mathbf{A}^{k})^\top \mathbf{C}^\top f_l\\ f_l^\top\mathbf{C}\mathbf{A}^k & \phi_{l,k}
        \end{bmatrix} \\
        &\ -\sum_{k=0}^\mu \sum_{l=1}^c \delta_{k,l}^{(i,j)} \begin{bmatrix}
            0 & (\mathbf{A}^{k})^\top \mathbf{C}^\top g_l\\ g_l^\top\mathbf{C}\mathbf{A}^k & b_l
        \end{bmatrix}\succeq 0,
        \end{split}
\end{equation*}
then $\mu$ satisfies~\eqref{eq:mu_containment}. 

\subsection{Proof of Proposition~\ref{prop:Rec-SMPC}} \label{app:reconditioning_proof}
\begin{proof}
    The proof is inspired by~\cite[Proposition~1]{wang2021recursive}. \par \textit{Part I: Recursive feasibility.} Consider the candidate solution
    \begin{equation*}
        \hat{\pi}_{t-k|k}(\cdot) = \pi^\star_{t-k+1|k-1}(w_{k-1}, \cdot),
    \end{equation*}
    i.e., it is the same as $\pi^\star_{t-k+1|k-1}$ but the first argument is fixed to the realized value of the disturbance $w_{k-1}$. Therefore,
    \begin{equation*}
        p(u_t| w_{0:k-1},\{\boldsymbol{\pi}^\mathrm{p}_k, \hat{\boldsymbol{\pi}}_k\}) = p(u_t | w_{0:k-1},\{\boldsymbol{\pi}^\mathrm{p}_{k-1}, \boldsymbol{\pi}_{k-1}^\star\}),
    \end{equation*}
    resulting in
    \begin{equation*}
        p(x_t| w_{0:k-1},\{\boldsymbol{\pi}^\mathrm{p}_k, \hat{\boldsymbol{\pi}}_k\}) = p(x_t | w_{0:k-1},\{\boldsymbol{\pi}^\mathrm{p}_{k-1}, \boldsymbol{\pi}_{k-1}^\star\}),
    \end{equation*}
    $\forall t \in \mathbb{I}_{\geq k}$. Consequently,
    \begin{align*}
         \mathrm{Pr}[(x_t,&\, u_t) \in \mathcal{C}^j | w_{0:k-1},\{\boldsymbol{\pi}^\mathrm{p}_k,\hat{\boldsymbol{\pi}}_k\}] =  \\ &= \mathrm{Pr}[(x_t, u_t) \in \mathcal{C}^j | w_{0:k-1},\{\boldsymbol{\pi}^\mathrm{p}_{k-1},\boldsymbol{\pi}_{k-1}^\star\}] = p^j_{t-k|k},
    \end{align*}
    $\forall j \in \mathbb{I}_{[1,c]}$, that is, the candidate solution satisfies all constraints in~\eqref{eq:Rec-SMPC}, proving recursive feasibility.
    \par \textit{Part II: Closed-loop chance constraint satisfaction.}
    Let $k \in \mathbb{I}_{> 0}$ and let $\mathbf{w}_k := w_{0:k-1}$. Then
    \begin{equation*}
        \begin{split}
            \mathrm{Pr}&[(x_t, u_t) \in \mathcal{C}^j | \{\boldsymbol{\pi}^\mathrm{p}_k, \boldsymbol{\pi}_k^\star\}]  \\&\stackrel{(*)}{=} \int \mathrm{Pr}[(x_t, u_t) \in \mathcal{C}^j | \mathbf{w}_k,\{\boldsymbol{\pi}^\mathrm{p}_k, \boldsymbol{\pi}_k^\star\}]\ p(\mathbf{w}_k)\ \mathrm{d}\mathbf{w}_k  \\
            &\stackrel{\eqref{eq:Rec-SMPC:constraints}}{\geq} \int p^j_{t-k|k}\ p(\mathbf{w}_k)\ \mathrm{d}\mathbf{w}_k \\
            &\stackrel{\eqref{eq:pjik}}{=} \int \mathrm{Pr}[(x_t, u_t) \in \mathcal{C}^j | \mathbf{w}_k,\{\boldsymbol{\pi}^\mathrm{p}_{k-1}, \boldsymbol{\pi}_{k-1}^\star\}]\ p(\mathbf{w}_k)\ \mathrm{d}\mathbf{w}_k  \\
            &\stackrel{(*)}{=} \mathrm{Pr}[(x_t, u_t) \in \mathcal{C}^j| \{\boldsymbol{\pi}^\mathrm{p}_{k-1}, \boldsymbol{\pi}_{k-1}^\star\}],
        \end{split}
    \end{equation*}
    $\forall j \in \mathbb{I}_{[1,c]},\ \forall t \in \mathbb{I}_{\geq k}$, where $(*)$ is the law of total probability. Applying this procedure recursively for $t=k$ 
    and using ~\eqref{eq:initial_prob} yields
    \begin{align*}
        \mathrm{Pr}[(x_k, u_k) \in \mathcal{C}^j] &= \mathrm{Pr}[(x_k, u_k) \in \mathcal{C}^j | \boldsymbol{\pi}_k^\mathrm{p}] \\ &\geq \mathrm{Pr}[(x_k, u_k) \in \mathcal{C}^j | \boldsymbol{\pi}_0^\star]  \geq p^j,\forall j \in \mathbb{I}_{[1,c]}. \qedhere
    \end{align*}
\end{proof}

\subsection{Proof of Theorem~\ref{thm:RHC}} \label{app:RHC_proof}
\begin{proof}
    \textit{Part I: Recursive feasibility.} Given the optimal solution of~\eqref{eq:RHC} at time step $k \in \mathbb{I}_{> 0}$ (or that of~\eqref{eq:DMC} at time $k = 0$), consider the candidate solution $\hat{\mathbf{z}}_k$, $\hat{\mathbf{v}}_k$, $\hat{\mathbf{\Phi}}^\mathrm{x}_{:|k}$, and $\hat{\mathbf{\Phi}}^\mathrm{u}_{:|k}$, defined in~\eqref{eq:hat_zvPsi_def}-\eqref{eq:hat_zNPsiN_def}. As the constraints of~\eqref{eq:RHC} are constructed based on this exact candidate sequence, they are satisfied with equality, proving recursive feasibility.
    \par \textit{Part II: Closed-loop chance constraint satisfaction.} 
    The proof follows the arguments in Part II of the proof of Proposition~\ref{prop:Rec-SMPC}, by proving that the proposed method ensures
    \begin{equation} \label{eq:CS_requirement}
        \mathrm{Pr}[(x^\star_{i|k}, u^\star_{i|k}) \in \mathcal{C}^j] \geq p^j_{i|k},\ \forall j \in \mathbb{I}_{[1,c]}.
    \end{equation}
    \par It is known that for $a \in \mathbb{R}$ and scalar Gaussian random variable $y$ with mean $\mu_\mathrm{y}$ and variance $\sigma_\mathrm{y}^2$,
    \begin{equation} \label{eq:Gaussian_CDF}
        \mathrm{Pr}[y \leq a] = \Phi\left(\frac{a - \mu_\mathrm{y}}{\sigma_\mathrm{y}}\right),
    \end{equation}
    where $\Phi$ denotes the cumulative distribution function of the standard normal distribution. Given the Gaussian disturbance~\eqref{eq:system}, it can be seen that 
    $p(x^\star_{i|k}, u^\star_{i|k})$ and $p(x^\star_{i+1|k-1}, u^\star_{i+1|k-1})$
    are Gaussian distributed, with means and variances as follows:
    \begin{align*}
        \mathbb{E}[[x_{i|k}^{\star\top}, u_{i|k}^{\star\top}]^\top] &= [z_{i|k}^{\star\top}, v_{i|k}^{\star\top}]^\top \\
         \mathrm{Var}[[x_{i|k}^{\star\top}, u_{i|k}^{\star\top}]^\top] &= \begin{bmatrix}
            \mathbf{\Phi}^{\mathrm{x}\star}_{i|k} \\ \mathbf{\Phi}^{\mathrm{u}\star}_{i|k}
        \end{bmatrix} \mathbf{\Sigma}^\mathrm{w}_i [*]^\top,
    \end{align*}
    $\forall i \in \mathbb{I}_{[0:N-1]}$, and in the tail of the horizon, 
    \begin{equation*}
        \mathbb{E}[[x_{N+i|k}^{\star\top}, u_{N+i|k}^{\star\top}]^\top] = [I_n, K^\top]^\top A_K^iz^\star_{N|k},
    \end{equation*}
    \begin{equation*}
    \begin{split}
        \mathrm{Var}[[&x_{N+i|k}^{\star\top}, u_{N+i|k}^{\star\top}]^\top] = \\ &\begin{bmatrix}
            I_n \\ K
        \end{bmatrix}
        \begin{bmatrix}
            A_K^i\mathbf{\Phi}^{\mathrm{x}\star}_{N|k} & I_n
        \end{bmatrix} \begin{bmatrix}
            \mathbf{\Sigma}^\mathrm{w}_N & \\ & \Sigma^\mathrm{x}_i
        \end{bmatrix} [*]^\top,
    \end{split}
    \end{equation*}
    $\forall i \in \mathbb{I}_{\geq 0}$. Similarly,
    \begin{align*}
        \mathbb{E}[[x^{\star\top}_{i+1|k-1} u^{\star\top}_{i+1|k-1}]^\top] &= [\hat{z}^\top_{i|k}, \hat{v}_{i|k}^\top]^\top \\
        \mathrm{Var}[[x_{i+1|k-1}^{\star\top}, u_{i+1|k-1}^{\star\top}]^\top] &= \begin{bmatrix}
            \hat{\mathbf{\Phi}}^\mathrm{x}_{i|k} \\ \hat{\mathbf{\Phi}}^\mathrm{u}_{i|k} 
        \end{bmatrix} \mathbf{\Sigma}^\mathrm{w}_i [*]^\top
    \end{align*}
    $\forall i \in \mathbb{I}_{[0:N-1]}$, and finally,
    \begin{equation*}
        \mathbb{E}[[x^{\star\top}_{N+i+1|k-1}, u^{\star\top}_{N+i+1|k-1}]^\top] = [I_n, K^\top]^\top A_K^i\hat{z}_{N|k}
    \end{equation*}
    \begin{equation*}
    \begin{split}
        \mathrm{Var}[[&x_{N+i+1|k-1}^{\star\top}, u_{N+i+1|k-1}^{\star\top}]^\top] = \\ 
        &\begin{bmatrix}
            I_n \\ K
        \end{bmatrix}
        \begin{bmatrix}
            A_K^i\hat{\mathbf{\Phi}}^\mathrm{x}_{N|k} & I_n
        \end{bmatrix} \begin{bmatrix}
            \mathbf{\Sigma}^\mathrm{w}_N & \\ & \Sigma^\mathrm{x}_i
        \end{bmatrix} [*]^\top
    \end{split}
    \end{equation*}
    $\forall i \in \mathbb{I}_{\geq 0}$. 
    \par First, consider $i \in \mathbb{I}_{[0,N-1]}$. Due to~\eqref{eq:Gaussian_CDF} combined with the fact that $\Phi$ is strictly monotonically increasing and using the derived formulae for the means and variances, we can see that~\eqref{eq:CS_requirement} is equivalent to
    \begin{equation*}
        \frac{b_j-G_j z^\star_{i|k} -H_j v^\star_{i|k}}{\norm{(\mathbf{\Sigma}^{\mathrm{w}}_i)^{1/2} \begin{bmatrix} (\mathbf{\Phi}^{\mathrm{x}\star}_{i|k})^\top & (\mathbf{\Phi}^{\mathrm{u}\star}_{i|k})^\top\end{bmatrix} \begin{bmatrix} G_j^\top \\ H_j^\top \end{bmatrix}}} \geq \alpha^j_{i|k};
    \end{equation*}
    in the non-degenerate case (i.e., when~\eqref{eq:pjik_degenerate} does not hold). In the degenerate case,~\eqref{eq:CS_requirement} vanishes whenever~\eqref{eq:previous_satisfied} does not hold, and otherwise it leads to the constraint set~\eqref{eq:Cijk1}. Thus,~\eqref{eq:RHC:constraints_N} ensures that~\eqref{eq:CS_requirement} holds $\forall k \in \mathbb{I}_{> 0}, \forall i \in \mathbb{I}_{[0,N-1]}$
    \par As for $i \in \mathbb{I}_{\geq N}$, similar arguments can be used to show that~\eqref{eq:CS_requirement} is equivalent to~\eqref{eq:SOC_k>0_tail}. Therefore,~\eqref{eq:RHC:constraints_terminal} implies~\eqref{eq:CS_requirement} $\forall k \in \mathbb{I}_{> 0}, \forall i \in \mathbb{I}_{\geq N}$. Furthermore,~\eqref{eq:DMC} ensures that~\eqref{eq:CS_requirement} holds $\forall i \in \mathbb{I}_{\geq 0}$ at $k=0$. Thus~\eqref{eq:CS_requirement} holds $\forall k \in \mathbb{I}_{\geq 0}, \forall i \in \mathbb{I}_{\geq 0}$ and hence closed-loop chance constraint satisfaction follows analogously to Proposition~\ref{prop:Rec-SMPC}.
    \par \textit{Part III: Bounded asymptotic average cost}.
    Let $J^\star_k$ denote the optimal cost of~\eqref{eq:DMC} and~\eqref{eq:RHC}, for $k=0$ and $k \in \mathbb{I}_{> 0}$, respectively, and define $\hat{J}_k$ as the cost associated with the candidate solution:
    \begin{align*}
        J^\star_k &= \sum_{i=0}^{N-1}\ell(z^\star_{i|k}, v^\star_{i|k}) \\
        &+ \sum_{i=1}^{N-1}(\mathrm{tr}(Q\mathbf{\Phi}^{\mathrm{x}\star}_{i|k}\mathbf{\Sigma}^\mathrm{w}_i(\mathbf{\Phi}^{\mathrm{x}\star}_{i|k})^\top) + \mathrm{tr}(R\mathbf{\Phi}^{\mathrm{u}\star}_{i|k}\mathbf{\Sigma}^\mathrm{w}_i(\mathbf{\Phi}^{\mathrm{u}\star}_{i|k})^\top))  \\ 
        &+ \ell_\mathrm{f}(z^\star_{N|k}) + \mathrm{tr}(P\mathbf{\Phi}^{\mathrm{x}\star}_{N|k}\mathbf{\Sigma}^\mathrm{w}_N(\mathbf{\Phi}^{\mathrm{x}\star}_{N|k})^\top),
    \end{align*}
    and
    \begin{align*}
        \hat{J}&_{k+1} = \sum_{i=0}^{N-1}\ell(\hat{z}_{i|k+1},\hat{v}_{i|k+1}) \\
        &+ \sum_{i=1}^{N-1}(\mathrm{tr}(Q\hat{\mathbf{\Phi}}^\mathrm{x}_{i|k+1}\mathbf{\Sigma}^\mathrm{w}_i\hat{\mathbf{\Phi}}^{\mathrm{x}\top}_{i|k+1}) + \mathrm{tr}(R\hat{\mathbf{\Phi}}^\mathrm{u}_{i|k+1}\mathbf{\Sigma}^\mathrm{w}_i\hat{\mathbf{\Phi}}^{\mathrm{u}\top}_{i|k+1}))  \\ 
        & + \ell_\mathrm{f}(\hat{z}_{N|k+1}) +\mathrm{tr}(P\hat{\mathbf{\Phi}}^\mathrm{x}_{N|k+1}\mathbf{\Sigma}^\mathrm{w}_N\hat{\mathbf{\Phi}}^{\mathrm{x}\top}_{N|k+1}).
    \end{align*}
    Then,
    \begin{align} \label{eq:proofs:beg_difference}
        \mathbb{E}_{w_k}[J^\star_{k+1} | w_{0:k-1}] - J^\star_k  \leq \mathbb{E}_{w_k}[\hat{J}_{k+1} | w_{0:k-1}] - J^\star_k.
    \end{align}
    Using~\eqref{eq:hat_zvPsi_def}-\eqref{eq:hat_zNPsiN_def}, it can be seen that
    \begin{align*}
        &\underset{w_k}{\mathbb{E}}[\ell(\hat{z}_{i|k+1}, \hat{v}_{i|k+1}) | w_{0:k-1}] = \ell(z^\star_{i+1|k}, v^\star_{i+1|k})\\&+ \mathrm{tr}(Q\Phi^{\mathrm{x}\star}_{i+1,1|k}\Sigma^\mathrm{w}(\Phi^{\mathrm{x}\star}_{i+1,1|k})^\top) + \mathrm{tr}(R\Phi^{\mathrm{u}\star}_{i+1,1|k}\Sigma^\mathrm{w}(\Phi^{\mathrm{u}\star}_{i+1,1|k})^\top),
    \end{align*}
    $\forall i \in \mathbb{I}_{[0,N-1]}$, and
    \begin{equation*}
    \begin{split}
        \underset{w_k}{\mathbb{E}}[\ell_\mathrm{f}(\hat{z}_{N|k+1})|w_{0:k-1}]& = \ell_\mathrm{f}(A_Kz^\star_{N|k}) \\ +& \mathrm{tr}(PA_K\Phi^{\mathrm{x}\star}_{N,1|k}\Sigma^\mathrm{w}(\Phi^{\mathrm{x}\star}_{N,1|k})^\top A_K^\top).
    \end{split}
    \end{equation*}
    Therefore, using the cyclic property and linearity of the trace operator,
    \begin{align}
        \mathbb{E}_{w_k}&[\hat{J}_{k+1} | w_{0:k-1}] - J^\star_k  = \mathrm{tr}(P\Sigma^\mathrm{w}) - \ell(z^\star_{0|k}, v^\star_{0|k}) \nonumber\\
        + &\mathrm{tr}((Q + K^\top R K + A_K^\top P A_K - P)\mathbf{\Phi}^{\mathrm{x}\star}_{N|k}\mathbf{\Sigma}^\mathrm{w}_N(\mathbf{\Phi}^{\mathrm{x}\star}_{N|k})^\top) \nonumber \\
        + &\nnorm{z^\star_{N|k}}_{Q + K^\top R K + A_K^\top P A_K - P}^2 \label{eq:proofs:end_difference} \\
        + &(r^\top K+q^\top - p_\mathrm{f}^\top(I_n-A_K))z_{N|k}^\star, \nonumber
    \end{align}
    where the last three terms vanish since $p_\mathrm{f}$ satisfies~\eqref{eq:p_Lyapunov} and $P$ satisfies the Lyapunov equation~\eqref{eq:P_Lyapunov}. Thus,~\eqref{eq:RHC:init},~\eqref{eq:RHC:control_law}, and~\eqref{eq:proofs:beg_difference} combined with~\eqref{eq:proofs:end_difference} results in the cost decrease
    \begin{equation} \label{eq:proofs:cost_decrease}
        \mathbb{E}_{w_k}[J^\star_{k+1} | w_{0:k-1}] - J^\star_k \leq\mathrm{tr}(P\Sigma^\mathrm{w}) - \ell(x_k,u_k).
    \end{equation}
    Following standard arguments from the SMPC literature using the law of total expectation and uniform lower boundedness of the cost (e.g.~\cite[Corollary~1]{hewing2020recursively}),~\eqref{eq:proofs:cost_decrease} leads to~\eqref{eq:lavg}. 
\end{proof}
    
\bibliographystyle{IEEEtran} 
\bibliography{Literature} 

@article{anderson2019system,
  title={System level synthesis},
  author={Anderson, James and Doyle, John C and Low, Steven H and Matni, Nikolai},
  journal={Annual Reviews in Control},
  volume={47},
  pages={364--393},
  year={2019},
  publisher={Elsevier}
}

@article{goulart2006optimization,
  title={Optimization over state feedback policies for robust control with constraints},
  author={Goulart, Paul J and Kerrigan, Eric C and Maciejowski, Jan M},
  journal={Automatica},
  volume={42},
  number={4},
  pages={523--533},
  year={2006},
  publisher={Elsevier}
}

@article{kouvaritakis2016model,
  title={Model predictive control},
  author={Kouvaritakis, Basil and Cannon, Mark},
  journal={Switzerland: Springer International Publishing},
  volume={38},
  pages={13--56},
  year={2016},
  publisher={Springer}
}

@article{sieber2021system,
  title={A system level approach to tube-based model predictive control},
  author={Sieber, Jerome and Bennani, Samir and Zeilinger, Melanie N},
  journal={IEEE Control Systems Letters},
  volume={6},
  pages={776--781},
  year={2021},
  publisher={IEEE}
}

@article{hewing2020recursively,
  title={Recursively feasible stochastic model predictive control using indirect feedback},
  author={Hewing, Lukas and Wabersich, Kim P and Zeilinger, Melanie N},
  journal={Automatica},
  volume={119},
  pages={109095},
  year={2020},
  publisher={Elsevier}
}

@article{li2021chance,
  title={Chance-constrained controller state and reference governor},
  author={Li, Nan and Girard, Anouck and Kolmanovsky, Ilya},
  journal={Automatica},
  volume={133},
  pages={109864},
  year={2021},
  publisher={Elsevier}
}

@article{gilbert1991linear,
  title={Linear systems with state and control constraints: The theory and application of maximal output admissible sets},
  author={Gilbert, Elmer G and Tan, K Tin},
  journal={IEEE Transactions on Automatic control},
  volume={36},
  number={9},
  pages={1008--1020},
  year={1991},
  publisher={IEEE}
}

@book{boyd1994linear,
  title={Linear matrix inequalities in system and control theory},
  author={Boyd, Stephen and El Ghaoui, Laurent and Feron, Eric and Balakrishnan, Venkataramanan},
  year={1994},
  publisher={SIAM}
}

@article{cannon2010stochastic,
  title={Stochastic tubes in model predictive control with probabilistic constraints},
  author={Cannon, Mark and Kouvaritakis, Basil and Rakovi{\'c}, Sa{\v{s}}a V and Cheng, Qifeng},
  journal={IEEE Transactions on Automatic Control},
  volume={56},
  number={1},
  pages={194--200},
  year={2010},
  publisher={IEEE}
}

@article{lorenzen2016constraint,
  title={Constraint-tightening and stability in stochastic model predictive control},
  author={Lorenzen, Matthias and Dabbene, Fabrizio and Tempo, Roberto and Allg{\"o}wer, Frank},
  journal={IEEE Transactions on Automatic Control},
  volume={62},
  number={7},
  pages={3165--3177},
  year={2016},
  publisher={IEEE}
}

@article{paulson2020stochastic,
  title={Stochastic model predictive control with joint chance constraints},
  author={Paulson, Joel A and Buehler, Edward A and Braatz, Richard D and Mesbah, Ali},
  journal={International Journal of Control},
  volume={93},
  number={1},
  pages={126--139},
  year={2020},
  publisher={Taylor \& Francis}
}

@inproceedings{farina2013probabilistic,
  title={A probabilistic approach to model predictive control},
  author={Farina, Marcello and Giulioni, Luca and Magni, Lalo and Scattolini, Riccardo},
  booktitle={Proc. 52nd IEEE conference on decision and control},
  pages={7734--7739},
  year={2013},
  organization={IEEE}
}

@inproceedings{hewing2018stochastic,
  title={Stochastic model predictive control for linear systems using probabilistic reachable sets},
  author={Hewing, Lukas and Zeilinger, Melanie N},
  booktitle={Proc. IEEE Conference on Decision and Control (CDC)},
  pages={5182--5188},
  year={2018},
  organization={IEEE}
}

@article{kohler2022recursively,
  title={Recursively feasible stochastic predictive control using an interpolating initial state constraint},
  author={K{\"o}hler, Johannes and Zeilinger, Melanie N},
  journal={IEEE Control Systems Letters},
  volume={6},
  pages={2743--2748},
  year={2022},
  publisher={IEEE}
}

@inproceedings{oldewurtel2008tractable,
  title={A tractable approximation of chance constrained stochastic {MPC} based on affine disturbance feedback},
  author={Oldewurtel, Frauke and Jones, Colin N and Morari, Manfred},
  booktitle={Proc. 47th IEEE Conference on Decision and Control (CDC)},
  pages={4731--4736},
  year={2008},
  organization={IEEE}
}

@article{mark2022recursively,
  title={Recursively feasible data-driven distributionally robust model predictive control with additive disturbances},
  author={Mark, Christoph and Liu, Steven},
  journal={IEEE Control Systems Letters},
  volume={7},
  pages={526--531},
  year={2022},
  publisher={IEEE}
}

@book{rawlings2017model,
  title={Model predictive control: theory, computation, and design},
  author={Rawlings, James Blake and Mayne, David Q and Diehl, Moritz and others},
  volume={2},
  year={2017},
  publisher={Nob Hill Publishing Madison, WI}
}

@article{farina2016stochastic,
  title={Stochastic linear model predictive control with chance constraints--a review},
  author={Farina, Marcello and Giulioni, Luca and Scattolini, Riccardo},
  journal={Journal of Process Control},
  volume={44},
  pages={53--67},
  year={2016},
  publisher={Elsevier}
}

@inproceedings{wang2021recursive,
  title={Recursive feasibility of stochastic model predictive control with mission-wide probabilistic constraints},
  author={Wang, Kai and Gros, S{\'e}bastien},
  booktitle={Proc. 60th IEEE Conference on Decision and Control (CDC)},
  pages={2312--2317},
  year={2021},
  organization={IEEE}
}

@article{osqp,
  author  = {Stellato, B. and Banjac, G. and Goulart, P. and Bemporad, A. and Boyd, S.},
  title   = {{OSQP}: an operator splitting solver for quadratic programs},
  journal = {Mathematical Programming Computation},
  year    = {2020},
  volume  = {12},
  number  = {4},
  pages   = {637--672},
  doi     = {10.1007/s12532-020-00179-2},
}

@article{diamond2016cvxpy,
  author  = {Steven Diamond and Stephen Boyd},
  title   = {{CVXPY}: {A} {P}ython-embedded modeling language for convex optimization},
  journal = {Journal of Machine Learning Research},
  year    = {2016},
  volume  = {17},
  number  = {83},
  pages   = {1--5},
}

@article{henderson1981vec,
  title={The vec-permutation matrix, the vec operator and Kronecker products: A review},
  author={Henderson, Harold V and Searle, Shayle R},
  journal={Linear and multilinear algebra},
  volume={9},
  number={4},
  pages={271--288},
  year={1981},
  publisher={Taylor \& Francis}
}

@inproceedings{pan2023data,
  title={Data-driven stochastic output-feedback predictive control: Recursive feasibility through interpolated initial conditions},
  author={Pan, Guanru and Ou, Ruchuan and Faulwasser, Timm},
  booktitle={Proc. Learning for Dynamics and Control Conference},
  pages={980--992},
  year={2023},
  organization={PMLR}
}

@article{mesbah2016stochastic,
  title={Stochastic model predictive control: An overview and perspectives for future research},
  author={Mesbah, Ali},
  journal={IEEE Control Systems Magazine},
  volume={36},
  number={6},
  pages={30--44},
  year={2016},
  publisher={IEEE}
}

@inproceedings{ono2012joint,
  title={Joint chance-constrained model predictive control with probabilistic resolvability},
  author={Ono, Masahiro},
  booktitle={Proc. American Control Conference (ACC)},
  pages={435--441},
  year={2012},
  organization={IEEE}
}

@article{arcari2023stochastic,
  title={Stochastic {MPC} with robustness to bounded parametric uncertainty},
  author={Arcari, Elena and Iannelli, Andrea and Carron, Andrea and Zeilinger, Melanie N},
  journal={IEEE Transactions on Automatic Control},
  year={2023},
  volume={68},
  number={12},
  pages={7601-7615},
  publisher={IEEE}
}

@article{mark2024stochastic,
  title={Stochastic {MPC} for Linear Systems With Unbounded Multiplicative Noise Guaranteeing Closed-Loop Chance Constraints Satisfaction},
  author={Mark, Christoph and Ravasio, Daniele and Farina, Marcello and G{\"o}rges, Daniel},
  journal={IEEE Control Systems Letters},
  year={2024},
  volume={8},
  number={},
  pages={2081-2086},
  publisher={IEEE}
}

@article{schluter2022stochastic,
  title={Stochastic model predictive control using initial state optimization},
  author={Schl{\"u}ter, Henning and Allg{\"o}wer, Frank},
  journal={IFAC-PapersOnLine},
  volume={55},
  number={30},
  pages={454--459},
  year={2022},
  publisher={Elsevier}
}

@article{li2024distributionally,
  title={Distributionally Robust Stochastic Data-Driven Predictive Control with Optimized Feedback Gain},
  author={Li, Ruiqi and Simpson-Porco, John W and Smith, Stephen L},
  journal={arXiv preprint arXiv:2409.05727},
  year={2024}
}

@article{knaup2024recursively,
  title={Recursively Feasible Stochastic Model Predictive Control for Time-Varying Linear Systems Subject to Unbounded Disturbances},
  author = {Knaup, Jacob W. and Tsiotras, Panagiotis},
  year = {2024},
  journal={arXiv preprint arXiv:2410.11107}
}

@inproceedings{prandini2012randomized,
  title={A randomized approach to stochastic model predictive control},
  author={Prandini, Maria and Garatti, Simone and Lygeros, John},
  booktitle={Proc. 51st IEEE Conference on Decision and Control (CDC)},
  pages={7315--7320},
  year={2012},
  organization={IEEE}
}

@manual{mosek,
   author = "MOSEK ApS",
   title = "The MOSEK Python Fusion API manual. Version 11.0.",
   year = 2025,
   url = "https://docs.mosek.com/latest/pythonfusion/index.html"
 }

@article{amrit2011economic,
  title={Economic optimization using model predictive control with a terminal cost},
  author={Amrit, Rishi and Rawlings, James B and Angeli, David},
  journal={Annual Reviews in Control},
  volume={35},
  number={2},
  pages={178--186},
  year={2011},
  publisher={Elsevier}
}

@inproceedings{lofberg2003approximations,
  title={Approximations of closed-loop minimax {MPC}},
  author={L{\"o}fberg, Johan},
  booktitle={42nd IEEE International Conference on Decision and Control (IEEE Cat. No. 03CH37475)},
  volume={2},
  pages={1438--1442},
  year={2003},
  organization={IEEE}
}

@article{leeman2025robust,
  title={Robust nonlinear optimal control via system level synthesis},
  author={Leeman, Antoine P and K{\"o}hler, Johannes and Zanelli, Andrea and Bennani, Samir and Zeilinger, Melanie N},
  journal={IEEE Transactions on Automatic Control},
  year={2025},
  publisher={IEEE}
}

\begin{IEEEbiography}[{\includegraphics[width=1in,height=1.25in,clip,keepaspectratio]{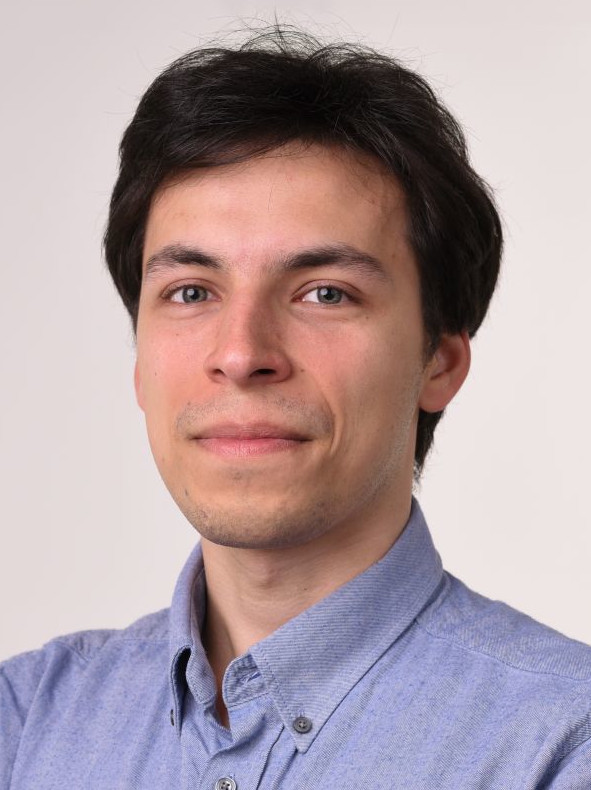}}]{Marcell Bartos}  \looseness -1 
 is a Ph.D. student under the supervision of Prof. Melanie Zeilinger and Prof. Florian Dörfler at ETH Zürich. He obtained his B.Sc. degree in Mechatronics Engineering from the Budapest University of Technology and Economics in 2021, and his M.Sc. degree in Robotics, Systems and Control at ETH Zürich in 2024. His areas of interest are model predictive control, adaptive control, and online learning for control.
\end{IEEEbiography}

\begin{IEEEbiography}[{\includegraphics[width=1in,height=1.25in,clip,keepaspectratio]{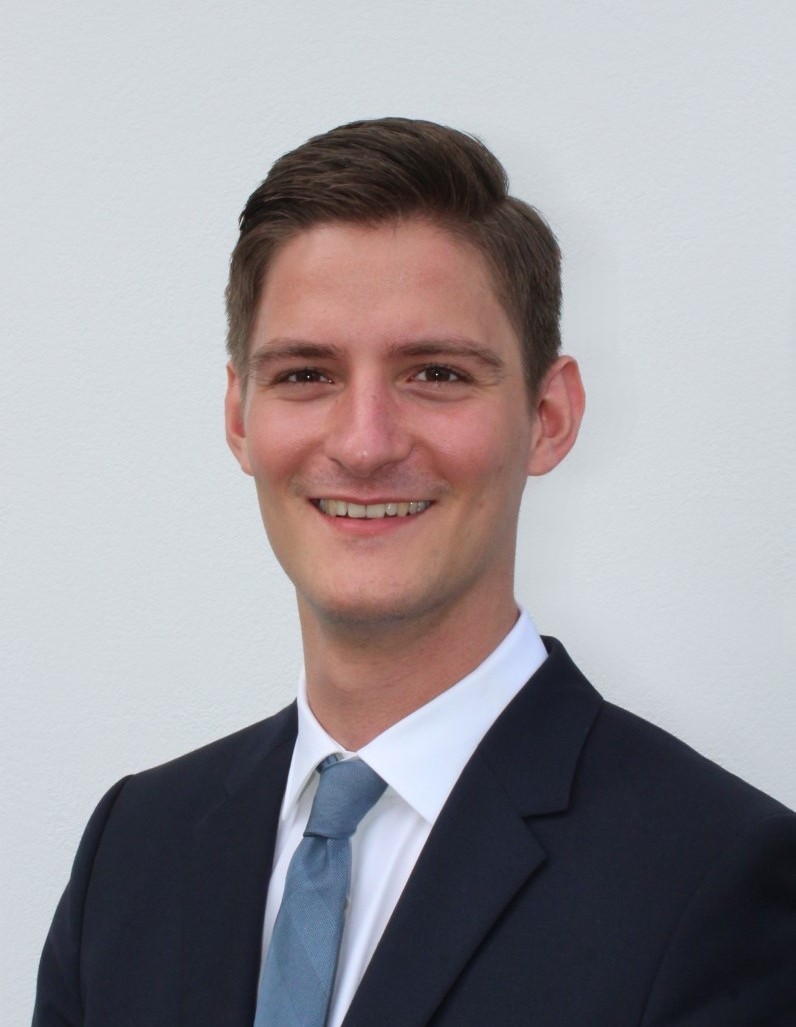}}]{Alexandre Didier} \looseness -1 is a Postdoctoral Researcher at the Institute for Dynamic Systems and Control (IDSC) at ETH Zurich. He received his MSc. in Robotics, Systems and Control from ETH Zurich in 2020 and his Ph.D. in Control Theory from ETH Zurich in 2025. His research interests lie especially in model predictive control and regret minimisation for safety-critical and uncertain systems.
\end{IEEEbiography}

\begin{IEEEbiography}[{\includegraphics[width=1in,height=1.25in,clip,keepaspectratio]{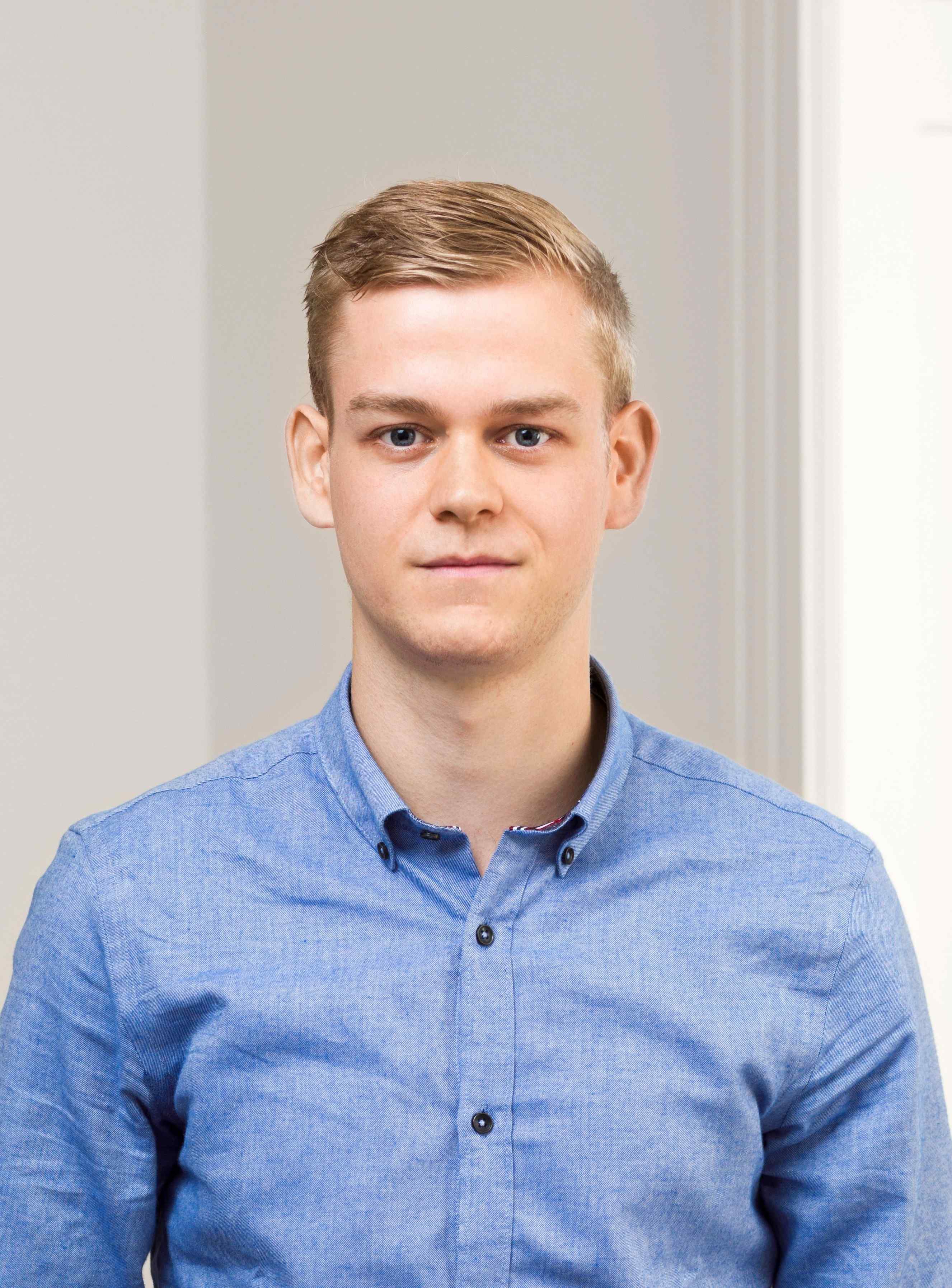}}]{Jerome Sieber} \looseness -1 is a Postdoctoral Researcher in Prof. Dr. Melanie Zeilinger's group at ETH Zurich, Switzerland. He received his MSc. in Robotics, Systems and Control from ETH Zurich in 2019, and his Ph.D. in Control Theory from ETH Zurich, in 2024. His research focuses on robust model predictive control, novel neural network architectures based on state-space representations, and control of deep sequence models.
\end{IEEEbiography}

\begin{IEEEbiography}[{\includegraphics[width=1in,height=1.25in,clip,keepaspectratio]{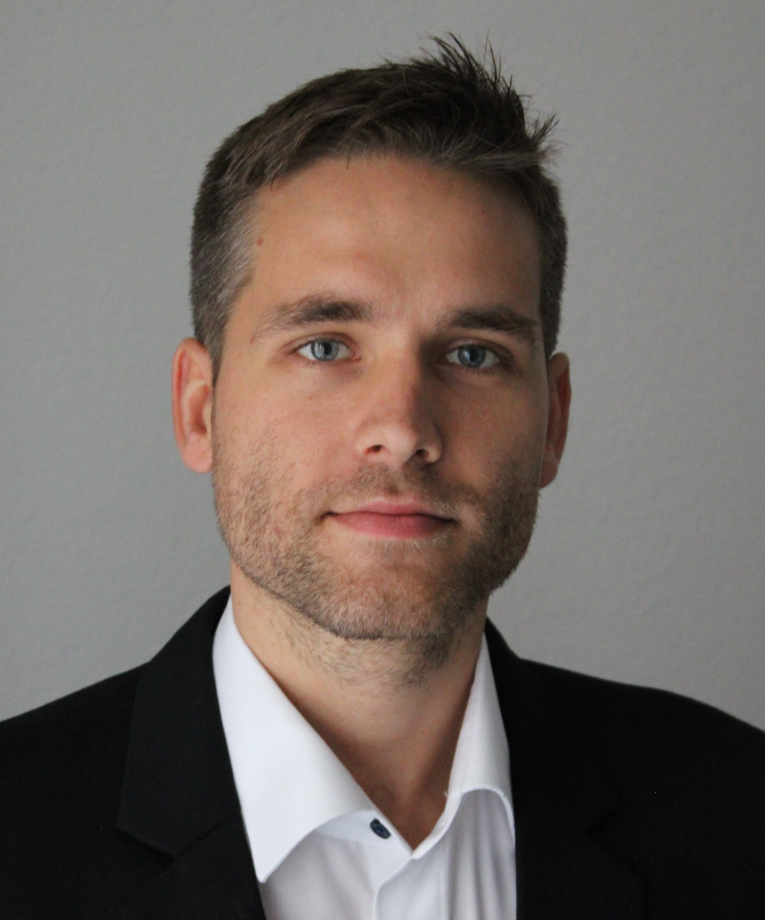}}]{Johannes K\"ohler} \looseness -1  is an Assistant Professor at Imperial College London. 
He received the Ph.D. degree from the University of Stuttgart, Germany, in 2021. 
From 2021 to 2025, he was a postdoctoral researcher at ETH Zurich, Switzerland.
He has received several awards including the 2021 European Systems \& Control PhD Thesis Award, the IEEE CSS George S. Axelby Outstanding Paper Award 2022, and the Journal of Process Control Paper Award 2023. 
His research interests include data-driven models and robust predictive control with applications to robotics, autonomous systems, and biomedical problems. 
\end{IEEEbiography}
\begin{IEEEbiography}[{\includegraphics[width=1in,height=1.25in,clip,keepaspectratio]{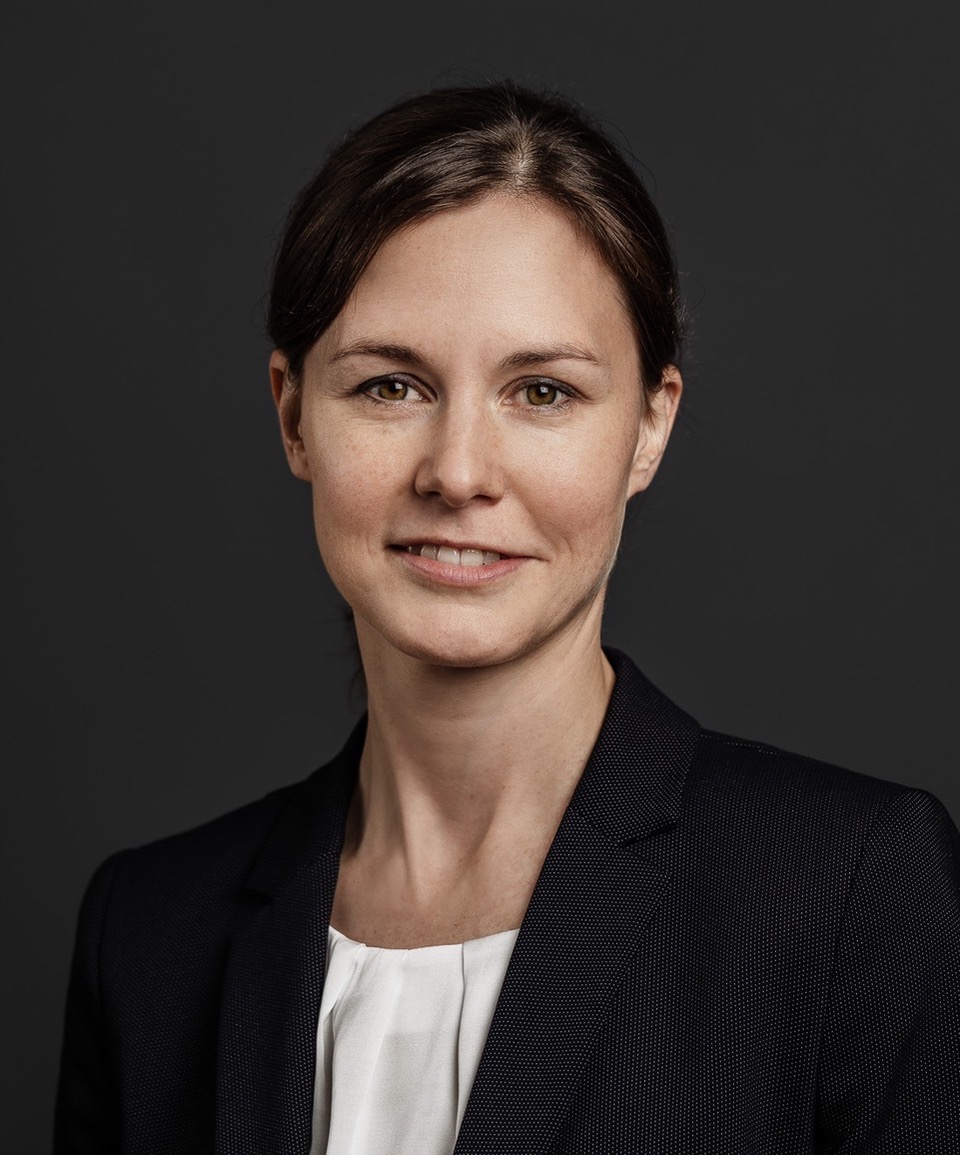}}]{Melanie N. Zeilinger}  \looseness -1 is an Associate Professor at ETH Zurich, Switzerland. She received the Diploma degree in engineering cybernetics from the University of Stuttgart, Germany, in 2006, and the Ph.D. degree with honors in electrical engineering from ETH Zurich, Switzerland, in 2011. From 2011 to 2012 she was a Postdoctoral Fellow with the Ecole Polytechnique Federale de Lausanne (EPFL), Switzerland. She was a Marie Curie Fellow and Postdoctoral Researcher with the Max Planck Institute for Intelligent Systems, Tübingen, Germany until 2015 and with the Department of Electrical Engineering and Computer Sciences at the University of California at Berkeley, CA, USA, from 2012 to 2014. From 2018 to 2019 she was a professor at the University of Freiburg, Germany. She was awarded the ETH medal for her PhD thesis, an SNF Professorship, the ETH Golden Owl for exceptional teaching in 2022 and the European Control Award in 2023. Her research interests include learning-based control with applications to robotics and human-in-the-loop control.
\end{IEEEbiography}
\end{document}